\pdfoutput=1
\documentclass[11pt,a4paper]{article}

\usepackage{authblk,secdot}

\usepackage[
bookmarks=true,
bookmarksnumbered=true,
colorlinks=true, pdfstartview=FitV, linkcolor=blue, citecolor=blue,
urlcolor=blue]{hyperref}

\usepackage{amsmath,amssymb,amsthm,amsfonts}
\usepackage{xcolor}
\usepackage{calc} 
\usepackage{mathtools}
\usepackage{bbm}
\usepackage{graphicx}
\usepackage{verbatim}
\usepackage{float}

\usepackage{subfig}

\usepackage[nameinlink,capitalize]{cleveref}
\usepackage[shortlabels]{enumitem}
\usepackage[normalem]{ulem}

\numberwithin{equation}{section}

\theoremstyle{plain}
\newtheorem{theorem}{Theorem}[section]

\newtheorem{proposition}[theorem]{Proposition}
\newtheorem{corollary}[theorem]{Corollary}

\theoremstyle{definition}

\newtheorem{assumption}[theorem]{Assumption}
\newtheorem{remark}[theorem]{Remark}

\newtheorem{innercustomgeneric}{\customgenericname}
\providecommand{\customgenericname}{}
\newcommand{\newcustomtheorem}[2]{%
	\newenvironment{#1}[1]
	{%
		\renewcommand\customgenericname{#2}%
		\renewcommand\theinnercustomgeneric{##1}%
		\innercustomgeneric
	}
	{\endinnercustomgeneric}
}

\newcustomtheorem{AssumptionA}{Assumption}

\newcommand{\C}{\mathcal{C}}
\newcommand{\F}{\mathcal{F}}
\renewcommand{\H}{\mathcal{H}}

\renewcommand{\L}{\mathcal{L}}
\newcommand{\N}{\mathbb{N}}
\newcommand{\PP}{\mathcal{P}}
\newcommand{\R}{\mathbb{R}}
\newcommand{\X}{\mathcal{X}}
\newcommand{\W}{\mathbb{W}}

\renewcommand{\d}{\partial}

\newcommand{\norm}[1]{\left\lVert#1\right\rVert} 

\DeclareMathOperator*{\argmin}{arg\, min}

\DeclareMathOperator{\Law}{Law}

\def\E{\hskip.15ex\mathsf{E}\hskip.10ex}
\def\P{\mathsf{P}}

\newcommand{\wt}{\widetilde}
\newcommand{\wh}{\widehat}
\renewcommand{\phi}{\varphi}
\newcommand{\eps}{\varepsilon}

\newcommand{\nn}{\nonumber}
\usepackage{color}

\makeatletter
\let\save@mathaccent\mathaccent
\newcommand*\if@single[3]{%
	\setbox0\hbox{${\mathaccent"0362{#1}}^H$}%
	\setbox2\hbox{${\mathaccent"0362{\kern0pt#1}}^H$}%
	\ifdim\ht0=\ht2 #3\else #2\fi
}
\newcommand*\rel@kern[1]{\kern#1\dimexpr\macc@kerna}
\newcommand*\widebar[1]{\@ifnextchar^{{\wide@bar{#1}{0}}}{\wide@bar{#1}{1}}}
\newcommand*\wide@bar[2]{\if@single{#1}{\wide@bar@{#1}{#2}{1}}{\wide@bar@{#1}{#2}{2}}}
\newcommand*\wide@bar@[3]{%
	\begingroup
	\def\mathaccent##1##2{%
		\let\mathaccent\save@mathaccent
		\if#32 \let\macc@nucleus\first@char \fi
		\setbox\z@\hbox{$\macc@style{\macc@nucleus}_{}$}%
		\setbox\tw@\hbox{$\macc@style{\macc@nucleus}{}_{}$}%
		\dimen@\wd\tw@
		\advance\dimen@-\wd\z@
		\divide\dimen@ 3
		\@tempdima\wd\tw@
		\advance\@tempdima-\scriptspace
		\divide\@tempdima 10
		\advance\dimen@-\@tempdima
		\ifdim\dimen@>\z@ \dimen@0pt\fi
		\rel@kern{0.6}\kern-\dimen@
		\if#31
		\overline{\rel@kern{-0.6}\kern\dimen@\macc@nucleus\rel@kern{0.4}\kern\dimen@}%
		\advance\dimen@0.4\dimexpr\macc@kerna
		\let\final@kern#2%
		\ifdim\dimen@<\z@ \let\final@kern1\fi
		\if\final@kern1 \kern-\dimen@\fi
		\else
		\overline{\rel@kern{-0.6}\kern\dimen@#1}%
		\fi
	}%
	\macc@depth\@ne
	\let\math@bgroup\@empty \let\math@egroup\macc@set@skewchar
	\mathsurround\z@ \frozen@everymath{\mathgroup\macc@group\relax}%
	\macc@set@skewchar\relax
	\let\mathaccentV\macc@nested@a
	\if#31
	\macc@nested@a\relax111{#1}%
	\else
	\def\gobble@till@marker##1\endmarker{}%
	\futurelet\first@char\gobble@till@marker#1\endmarker
	\ifcat\noexpand\first@char A\else
	\def\first@char{}%
	\fi
	\macc@nested@a\relax111{\first@char}%
	\fi
	\endgroup
}
\makeatother
\begin{document}

\title{A Reproducing Kernel Hilbert Space approach to singular local stochastic volatility McKean-Vlasov models}

\author[1]{Christian Bayer\thanks{\texttt{christian.bayer@wias-berlin.de}}}
\author[2]{Denis Belomestny\thanks{\texttt{denis.belomestny@uni-due.de}}}
\author[1]{Oleg Butkovsky\thanks{\texttt{oleg.butkovskiy@gmail.com}}}
\author[1]{John~Schoenmakers\thanks{\texttt{john.schoenmakers@wias-berlin.de}}}

\affil[1]{Weierstrass Institute, Mohrenstrasse 39, 10117 Berlin, Germany.}
\affil[2]{Duisburg-Essen University, Essen}

\date{January 12, 2024}
\maketitle

\begin{abstract}
Motivated by the challenges related to the calibration of financial models, we consider the problem of numerically solving a singular McKean-Vlasov equation
$$
d X_t= \sigma(t,X_t) X_t \frac{\sqrt v_t}{\sqrt {\E[v_t|X_t]}}dW_t,
$$
where $W$ is a Brownian motion and $v$ is an adapted diffusion process. This equation can be considered as a singular local stochastic volatility model.
Whilst such models are quite
popular among practitioners, unfortunately, its well-posedness has not been fully understood yet and, in general, is possibly not guaranteed at all.
We develop a novel regularization approach based on the reproducing kernel Hilbert space (RKHS) technique and show that  the regularized  model is well-posed.  Furthermore, we prove propagation of chaos. We demonstrate numerically  that a thus regularized model is able to perfectly replicate option prices due to typical local volatility models. Our results are also applicable to more general McKean--Vlasov equations.
\end{abstract}

\section{Introduction}
\label{sec:introduction}

The present article is motivated by \cite{G-HL}, wherein Guyon and Henry-Labord\`ere proposed a particle method for the calibration of local stochastic volatility  models (e.g. stock price models). For ease of presentation,
let us assume zero interest rates and recall that \emph{local volatility models}
\begin{equation}
	dX_t = \sigma(t,X_t) X_t dW_t \label{dup},
\end{equation}
where $W$ denotes a one-dimensional Brownian motion under a risk-neutral measure and $X$ the price of a stock, can replicate any sufficiently regular implied volatility surface, provided that we choose the local volatility according to \emph{Dupire's formula}, symbolically, $\sigma \equiv \sigma_{\text{Dup}}$ \cite{Dup}. (In case of deterministic nonzero interest rates the discussion below
remains virtually unchanged after passing to forward stock and option prices). Unfortunately, it is well understood that Dupire's model exhibits unrealistic random price behavior despite perfect fits to market prices of options. On the other hand, \emph{stochastic volatility models}
\begin{equation*}
	dX_t = \sqrt{v_t} X_t dW_t
\end{equation*}
for a suitably chosen stochastic variance process $v_t$, may lead to realistic (in particular, time-homogeneous) dynamics, but are typically difficult or impossible to fit to observed implied volatility surfaces. We refer to \cite{gatheral2011volatility} for an overview of stochastic and local volatility models.
Local stochastic volatility models can combine the advantages of both local and stochastic volatility models. Indeed, if the stock price is given by
\begin{equation*}
	dX_t = \sqrt{v_t} \sigma(t,X_t) X_t dW_t,
\end{equation*}
then it exactly fits the observed market option prices provided that
\begin{equation*}
	\sigma_{\text{Dup}}(t,x)^2 = \sigma(t,x)^2 \E\left[ v_t | X_t = x \right].
\end{equation*}
This is a simple consequence of the celebrated Gy\"ongy's Markovian projection theorem \cite[Theorem~4.6]{Gyo},  see also \cite[Corollary~3.7]{BS13}. With this choice of $\sigma$ we have
\begin{equation}
	\label{eq:stochastic-local-volatility-upd}
	dX_t = \sigma_{\text{Dup}}(t,X_t) X_t\frac{\sqrt{v_t} }{\sqrt{\E\left[ v_t | X_t \right]}} dW_t,
\end{equation}
Note that  $v$ in \eqref{eq:stochastic-local-volatility-upd} can be any integrable  and positive adapted stochastic process. In a sense, (\ref{eq:stochastic-local-volatility-upd}) may be considered as an inversion of the Markovian projection due to \cite{Gyo}, applied to Dupire's local volatility model, i.e. (\ref{dup}) with $\sigma \equiv \sigma_{\text{Dup}}$.
\par
Thus, the stochastic local volatility model of McKean--Vlasov type \eqref{eq:stochastic-local-volatility-upd} solves the smile calibration problem. However, equation   \eqref{eq:stochastic-local-volatility-upd} is singular in a sense explained below and very hard to analyze and to solve. Even the problem of proving existence or uniqueness for \eqref{eq:stochastic-local-volatility-upd} (under various assumptions on $v$) turned out to be notoriously difficult and only a few results are available; we refer to \cite{SHKOL} for an extensive discussion and literature review. Let us recall that the theory of standard McKean--Vlasov equations of the form
\begin{equation*}
	d Z_t = \widetilde{H}\left(t, Z_t, \mu_t \right) d t + \widetilde{F}\left(t, Z_t, \mu_t \right) d W_t
\end{equation*}
with $\mu_t = \mathrm{Law}(Z_t)$, is well understood under appropriate regularity conditions, in particular, Lipschitz continuity of $\wt H$ and $\wt F$ w.r.t.~the standard Euclidean distances in the first two arguments and w.r.t.~the Wasserstein distance in~$\mu_t$, see \cite{funaki1984certain,CaDe1,MV16}. Denoting $Z_t \coloneqq (X_t, Y_t)$, it is not difficult to see that the conditional expectation $(x, \mu_t) \mapsto \E\left[ A(Y_t) \mid X_t = x \right]$ is, in general, not Lipschitz continuous in the above sense. Therefore, the standard theory does not apply to 
\eqref{eq:stochastic-local-volatility-upd}.

There are a number of results available in the literature where the Lipschitz condition on drift and diffusion is not imposed.
Bossy and Jabir \cite{bossy2017wellposedness} considered singular McKean-Vlasov (MV) systems of the form:
\begin{align}
	dX_t &= \E[\ell(X_t) | Y_t] dt + \E[\gamma(X_t) | Y_t] dW_t,\label{eq:bossy-jabir-1-b}\\
	dY_t &= b(X_t, Y_t) dt + \sigma(Y_t) dB_t\label{eq:bossy-jabir-1-bb},
\end{align}
or, alternatively, the seemingly even less regular equation
\begin{equation}
	\label{eq:bossy-jabir-2}
	dX_t = \sigma(p(t, X_t)) dW_t,
\end{equation}
where $p(t, \cdot)$ denotes the density of $X_t$. \cite{bossy2017wellposedness} establishes well-posedness of~\eqref{eq:bossy-jabir-1-b}--\eqref{eq:bossy-jabir-1-bb} and~\eqref{eq:bossy-jabir-2} under suitable regularity conditions (in particular, ellipticity) based on energy estimates of the corresponding non-linear PDEs. Interestingly, these techniques break down when the roles of $X$ and $Y$ are reversed in~\eqref{eq:bossy-jabir-1-b}--\eqref{eq:bossy-jabir-1-bb}, that is, when $\E[\gamma(X_t) | Y_t]$ is replaced by $\E[\gamma(Y_t) | X_t]$ in~\eqref{eq:bossy-jabir-1-b} -- and similarly for the drift term. Hence, the results of \cite{bossy2017wellposedness} do not imply well-posedness of \eqref{eq:stochastic-local-volatility-upd}.
In \cite{SHKOL}, Lacker, Shkolnikov, and Zhang  studied the following two-dimensional SDE,
\begin{align} dX_t&=b_1(X_t)\frac{h(Y_t)}{\E[h(Y_t)|X_t]}\,dt+\sigma_1(X_t)\frac{f(Y_t)}{\sqrt{\E[f^2(Y_t)|X_t]}}\, dW_t,\label{eq1Shkol}
	\\
	dY_t&=b_2(Y_t)\,dt+\sigma_2(Y_t)\,dB_t,\label{eq2Shkol}
\end{align}
where $W$ and $B$ are two independent one-dimensional Brownian motions. Clearly, this can be seen as a generalization of \eqref{eq:stochastic-local-volatility-upd} with a non-zero drift and with the process $v$ chosen in a special way.
The authors proved  strong existence and uniqueness of the solutions to~\eqref{eq1Shkol}--\eqref{eq2Shkol} in the \emph{stationary} case. In particular, this implies strong conditions on $b_1$ and $b_2$, but also requires the initial value $(X_0,Y_0)$ to be random and to have the stationary distribution. Existence and uniqueness of \eqref{eq1Shkol}--\eqref{eq2Shkol} in the general case (without the stationarity assumptions) remains open. Finally, let us mention the result of Jourdain and Zhou \cite[Theorem~2.2]{jourdain2020existence}, which established weak existence of the solutions to \eqref{eq:stochastic-local-volatility-upd} for the case when $v$ is a jump process taking finitely many values.
\par
Another question apart from well-posedness of these singular McKean--Vlasov equations is how to solve them numerically (in a certain sense). Let us recall that even for standard SDEs with singular or irregular drift, where existence/uniqueness is known for quite some time, the convergence of the corresponding Euler scheme with non-vanishing rate has been established only very recently \cite{BDG,Men}. The situation with the  singular McKean--Vlasov equations presented above is much more complicated and very few results are available in the literature.
In particular, the results of \cite{SHKOL} do not provide a way to construct a numerical algorithm for solving \eqref{eq:stochastic-local-volatility-upd}   even in the stationary case considered there.
\par
In this paper, we study the problem of  numerically solving singular McKean-Vlasov (MV)  equations of a more general form than \eqref{eq:stochastic-local-volatility-upd}:
\begin{equation}
	\label{eq:X-dynamics}
	d X_t = H\left(t, X_t, Y_t, \E[A_1(Y_t) | X_t] \right) d t + F\left(t, X_t, Y_t, \E[A_2(Y_t) | X_t] \right) d W_t,
\end{equation}
where $H, F, A_1, A_2$ are sufficiently regular functions, $W$ is a $d$-dimensional Brownian motion, and $Y$ is a given stochastic process, for example, a diffusion process. A  key issue is how to approximate the conditional expectations
$\E[A_i(Y_t) | X_t = x]$, $i=1,2$, $x\in\R^d$.
\par
Guyon and Henry-Labord\`ere in the seminal paper \cite{G-HL}  suggested an approach to tackle this problem (see also \cite{antonelli2002rate}). They used the ``identity''
\begin{equation*}
	\E[A(Y_t) | X_t = x] \text{``$=$''} \frac{\E A(Y_t) \delta_x(X_t)}{\E \delta_x(X_t)},
\end{equation*}
where $\delta_x$ is the Dirac delta function concentrated at $x$. This suggests the following approximation:
\begin{equation}
	\label{eq:nadaraya-watson}
	\E[A(Y_t) | X_t = x] \approx \frac{\sum_{i=1}^N A(Y^{i,N}_t) k_\varepsilon(X^{i,N}_t - x)}{\sum_{i=1}^N k_\varepsilon(X^{i,N}_t - x)}.
\end{equation}
Here $(X^{i,N}, Y^{i,N})_{i=1\hdots N}$ is a particle system,  $k_\varepsilon(\cdot) \approx \delta_0(\cdot)$ is a regularizing kernel, and $\eps>0$ is a small parameter. This technique for solving \eqref{eq:X-dynamics} (assuming \eqref{eq:X-dynamics} has a solution for a moment) works very well in practice, especially when coupled with interpolation on a grid in $x$-space. Due to the local nature of the regression performed, the method can be justified under only weak regularity assumptions on the conditional expectation -- note, however, that the interpolation part might require higher order regularity.

On the other hand, the method has an important disadvantage shared by all local regression methods: For any given point $x$, only points $(X^i,Y^i)$ in a neighborhood around $x$ of size proportional to $\eps$ contribute to the estimate of $\E[A(Y_t) | X_t = x]$ -- as $k_\varepsilon(\cdot - x)\approx 0$ outside that neighborhood. Hence, local regression cannot take advantage of ``global'' information about the structure of the function $x \mapsto \E[A(Y_t) | X_t = x]$. If, for example, the conditional expectation can be globally approximated via a  polynomial, it is highly inefficient (from a computational point of view) to approximate it locally using~\eqref{eq:nadaraya-watson}. Taken to the extremes, if we assume a compactly supported kernel $k$ and formally take $\eps = 0$, then the estimator~\eqref{eq:nadaraya-watson} collapses to $\E[A(Y_t) | X_t = X^{i,N}_t] \approx A(Y^{i,N}_t)$, since only $X^{i,N}_t$ is close enough to itself to contribute to the estimator. In the context of the stochastic local volatility model~\eqref{eq:stochastic-local-volatility-upd}, this means that the dynamics silently collapses to a pure local volatility dynamics, if $\eps$ is chosen too small.

This disadvantage of local regression methods can be avoided by using \emph{global regression} techniques. Indeed, taking advantage of global regularity and global structural features of the unknown target function, global regression methods are often seen to be more efficient than their local counterparts, see, e.g., \cite{fbsite}. On the other hand, the global regression methods require more regularity (e.g. global smoothness) than the minimal assumptions needed for local regression methods. In addition, the choice of basis functions can be crucial for global regression methods.

In fact, the starting point of this work was to replace~\eqref{eq:nadaraya-watson} by global regression based on, say, $L$ basis functions. However, it turns out that Lipschitz constants of the resulting approximation to the conditional expectations in terms of the particle distribution explode as $L \to \infty$, unless the basis functions are carefully chosen.

As an alternative to \cite{G-HL} we propose in this paper a novel approach based on ridge regression in the context of reproducing kernel Hilbert spaces (RKHS) which, in particular, does not have either of the above mentioned disadvantages, even when the number of basis functions is infinite.

Recall that an RKHS  $\H$ is a Hilbert space of real valued functions
${f: \X \rightarrow \R}$,
such that the evaluation map $\H \ni f \mapsto f(x)$ is continuous for every $x\in \X$. This crucial property implies that there exists a positive symmetric kernel $k: \X\times\X$ $\rightarrow \R$, i.e. for any $c_{1},...,c_{n}\in\mathbb{R},$ $x_{1},...,x_{n}\in\X$ one
has%
\[
\sum_{i,j=1}^{n}c_{i}c_{j}k(x_{i},x_{j})\geq0,
\]
such that for every $x\in \X,$ $k_x$ $:=$ $k(\cdot,x)$ $\in$ $\H$,  and one has that
$\left<f,k_x\right>_\H$ $=$ $f(x)$, {for all}  $f\in\H.$ As a main feature, any positive definite kernel $k$ uniquely determines a RKHS $\H$
and the other way around. In our setting we will consider $\X\subset \R^d$.
For a detailed introduction and further properties of RKHS we refer to the literature, for example \cite[Chapter~4]{steinwart2008support}.
We recall that the RKHS framework is popular in machine learning where it is widely used for computing conditional expectations. In the learning context, kernel methods are most prominently used in order to avoid the curse of dimensionality when dealing with high-dimensional features by the \emph{kernel trick}. We stress that this issue is not relevant in the application to calibration of equity models -- but might be interesting for more general, high-dimensional singular McKean-Vlasov systems.

Consider a pair of random variables $(X,Y)$ taking values in $\X\times\X$ with finite second moments and denote $\nu \coloneqq \Law(X,Y)$. Suppose that $A\colon\mathcal{X} \to \R$ is sufficiently regular and  $\H$ is large enough so that we have $\E\left[ A(Y) | X = \cdot \right] \in \mathcal{H}$.
Then, formally,%
\begin{align*}
	c_{A}^{\nu}(\cdot)\coloneqq\int_{\mathcal{X}\times\mathcal{X}}k(\cdot
	,x)A(y)\nu(dx,dy)  & =\int_{\mathcal{X}}k(\cdot,x)\nu(dx,\mathcal{X}%
	)\int_{\mathcal{X}}A(y)\nu(dy|x)\\
	& =\int_{\mathcal{X}}k(\cdot,x)\E\left[  A(Y)|X=x\right]  \nu(dx,\mathcal{X}%
	)\\
	& =:\mathcal{C}^{\nu}\E\left[  A(Y)|X=\cdot\right],
\end{align*}
where
\begin{equation*}
	\mathcal{C}^{\nu}f\coloneqq\int_{\mathcal{X}}k(\cdot,x)f(x)\nu(dx,\mathcal{X}), \quad f\in\H.
\end{equation*}
Unfortunately, in general,  the operator $\mathcal{C}^\nu$ is not invertible. As $\mathcal{C}^\nu$ is positive definite, it is, however, possible to \emph{regularize} the inversion by replacing $\mathcal{C}^\nu$ by $\mathcal{C}^\nu + \lambda I_{\mathcal{H}}$ for some $\lambda > 0$ where \(I_{\mathcal{H}}\) is the identity operator on \(\mathcal{H}\). 
Indeed, it turns out that
\begin{equation}
	m^\lambda_A(\cdot;\nu) :=(\mathcal{C}^\nu + \lambda I_{\mathcal{H}})^{-1} c^\nu_A,\label{mla0}
\end{equation}
is the solution to the minimization problem
\begin{equation}\label{mla}
	m^{\lambda}_A(\cdot;\nu):=\argmin_{f\in \H}\bigl(\E (A(Y)-f(X))^2+\lambda\| f\|_{\mathcal{H}}^{2}\bigr),
\end{equation}
see \cref{prop: apcon}.
On the other hand one also has
\begin{equation*}
	\E[ A(Y) | X = \cdot] = \argmin_{f\in L_2 (\R^d,\Law(X))}\E (A(Y)-f(X))^2,
\end{equation*}
and therefore it is natural to expect that if
$\lambda>0$ is small enough and $\mathcal{H}$ is large enough, then  $m^{\lambda}_A(\cdot; \nu)\approx\E[A(Y) | X = \cdot],$ that is, $m^{\lambda}_A(\cdot; \nu)$ is close to the true conditional expectation.
\par
The main result of the article is that the regularized MV system obtained by replacing the conditional expectations with their regularized versions (\ref{mla0}) in~\eqref{eq:X-dynamics} is well-posed and propagation of chaos holds for the corresponding particle system, see \cref{prop:exist-reg} and \cref{T:MR2}.
To establish these theorems, we study the  joint regularity of $m^{\lambda}_A(x;\nu)$ in the space variable $x$, and the measure $\nu$ for fixed $\lambda>0$. These
type of results are almost absent in the literature on RKHS and we here fill this gap. In particular, we prove that under suitable conditions, $m^{\lambda}_A(x; \nu)$ is Lipschitz in both arguments, that is, w.r.t.~the standard Euclidean norm in $x$ and the Wasserstein-$1$-norm in $\nu$,  and, can be calculated numerically in an efficient way, see Section~\ref{sec:mkv}.
Additionally, in Section~\ref{sec:rkhs} we study the convergence of $m^{\lambda}_A(\cdot; \nu)$ in (\ref{mla0}) to the true conditional expectation for  fixed $\nu$ as $\lambda\searrow0$ .

Let us note that, as a further nice feature of the RKHS approach compared to
the kernel method of \cite{G-HL}, one may incorporate, at least
in principle, global prior information concerning properties of
$\E[ A(Y) | X = \cdot]$ into the choice of the RKHS generating kernel $k$. In
a nutshell, if one anticipates beforehand that $\E[ A(Y) | X = x]\approx
f(x)$ for some known ``nice'' function $f,$ one may
pass on to a new kernel $\widetilde{k}(x,y) \coloneqq k(x,y)+f(x)f(y).$ This degree of
freedom is similar to, for example, the possibility of choosing basis
functions in line with the problem under consideration in the usual regression
methods for American options.
We also note that the Lipschitz constants for $m^\lambda_A(\cdot;\nu)$ with respect to both arguments
are expressed in bounds related to $A$ and the kernel $k$, only, see  \cref{lip}. In contrast, if we would have dealt with standard ridge regression, that is, ridge regression based on a fixed system of basis functions, we would have to impose restrictions on the regression coefficients leading to a nonconvex constrained optimization problem.

Thus, the contribution of the current work is fourfold. First, we propose a  RKHS-based approach to regularize \eqref{eq:X-dynamics} and prove the well-posedness of the regularized equation. Second, we show convergence  of the approximation
\eqref{mla} to the true conditional expectation as $\lambda\searrow0$. Third, we suggest a particle based approximation of  the regularized equation and analyze its convergence. Finally, we apply our algorithm to the problem of smile calibration in finance and illustrate its performance on simulated data.
In particular, we validate our results by solving numerically a regularized version of \eqref{eq:stochastic-local-volatility-upd} (with $m^{\lambda}_A$ in place of the conditional expectation). We show that our system
is indeed an approximate solution to \eqref{eq:stochastic-local-volatility-upd} in the sense that we get very close fits of the implied volatility surface --- the final goal of the smile calibration problem.

\medskip
The rest of the paper is organized as follows. Our main theoretical results are given in \cref{sec:mkv}. Convergence properties of the regularized conditional expectation $m^\lambda_A$ are established in \cref{sec:rkhs}. A numerical algorithm for solving \eqref{eq:X-dynamics} and an efficient implementable approximation of $m^\lambda_A$ are discussed in \cref{sec:alg}.  \cref{sec:app} contains numerical examples. The results of the paper are summarized in \cref{S:6}. Finally, all the proofs are placed in \cref{sec:proof}.

\medskip
\noindent\textbf{Convention on constants.}
Throughout the paper $C$ denotes a positive constant whose value may change from line to line. The dependence of constants on parameters if needed will be indicated, e.g, $C(\lambda)$.

\medskip
\noindent\textbf{Acknowledgements.}
	The authors would like to thank the referees and the associated editor for their helpful comments and feedback. We are also grateful to  Peter Friz and Mykhaylo Shkolnikov for useful discussions. OB would like to thank Vadim Sukhov for very helpful conversations regarding implementing parallelized algorithms in Python. DB acknowledges the financial support from Deutsche Forschungsgemeinschaft (DFG), Grant Nr.497300407. 
	CB, OB, and JS are supported by the DFG Research Unit FOR 2402. OB is funded by the Deutsche Forschungsgemeinschaft (DFG, German Research Foundation) under Germany's Excellence Strategy --- The Berlin Mathematics Research Center MATH+ (EXC-2046/1, project ID: 390685689, sub-project EF1-22).

\section{Main results}
\label{sec:mkv}

We begin by introducing the basic notation. For $a\in\R,$ we denote $a^+:=\max(a,0)$.
Let $(\Omega, \F,\P)$ be a probability space. For $d\in\N$, let $\X\subset\R^d$ be an open subset, and $\PP_2(\X)$ be the set of all probability measures on $(\X,\mathcal{B}(\X))$ with finite second moment. If $\mu,\nu\in\PP_2(\X) $, $p\in[1,2]$, then we denote the \textit{Wasserstein-p} (Kantorovich) distance between them by
\begin{equation*}
	\W_p(\mu,\nu):=\inf (\E|X-Y|^p)^{1/p},
\end{equation*}
where the infimum is taken over all random variables $X,Y$ with $\Law(X)=\mu$ and ${\Law(Y)=\nu}$.
Let $k:\X\times\X\to\R$ be a symmetric, positive definite kernel, and $\mathcal{H}$ be a reproducing kernel Hilbert space of
functions $f\colon\X\to\R$
associated with the kernel $k$. That is, for any  $x\in\X$,  $f\in\H$ one has
\begin{equation*}
	f(x)=\langle f,k(x,\cdot)\rangle_{\H}.
\end{equation*}
In particular, $\langle k(x,\cdot),k(y,\cdot)\rangle_{\H}=k(x,y)$, for any $x,y\in\X$. We refer to \cite[Chapter~4]{steinwart2008support} for further properties of RKHS.

Let  $A\colon \X\to\R$ be a measurable function such that $|A(x)|\le C(1+|x|)$ for some universal constant $C>0$ and all $x\in\X$. For $\nu\in \PP_{2}(\X\times\X)$, $\lambda\ge0$ consider the following optimization problem (\textit{ridge regression})
\begin{equation}\label{minv}
	m^\lambda_A(\cdot; \nu) := \argmin_{f\in\H}\left\{  \int_{\X\times\X}|A(y)-f(x)|^2\,\nu(dx,dy)+\lambda
	\| f\|_{\H}^2\right\}.
\end{equation}
We fix $T>0$, $d\in\N$ and consider the system
\begin{align}
	\label{eq:mvsde-1}
	dX_t &= H(t,X_t, Y_t, \E[A_1(Y_t) | X_t] ) d t
	\!+ F(t,X_t, Y_t, \E[A_2(Y_t) | X_t] ) d W_t^X
	\\
	\label{eq:mvsde-2}
	dY_t&=b(t,Y_t)dt+\sigma(t,Y_t)dW_t^Y,
\end{align}
where $H\colon[0,T]\times\R^d\times \R^d \times \R\to \R^d$,
$F\colon [0,T]\times \R^d\times \R^d \times \R\to \R^d\times  \R^d$,
$A_i\colon \R^d\to \R$, $b\colon[0,T]\times\R^d\to \R^d$, $\sigma\colon[0,T]\times\R^d\to \R^d\times \R^d$ are measurable functions, $W^X,W^Y$ are two (possibly correlated) $d$-dimensional Brownian motions on $(\Omega, \F,\P)$, and $t\in[0,T]$. We note that our choice of $Y$ as a diffusion process in \eqref{eq:mvsde-2} is mostly for convenience, and we expect our results to hold in more generality, when appropriately modified.

Denote $\mu_t:=\Law(X_t,Y_t)$. As mentioned above, the functional  
$$(x, \mu_t) \mapsto \E\left[ A_i(Y_t) | X_t = x \right]$$ 
is not Lipschitz continuous even if $A_i$ is smooth. Therefore the classical results on well-posedness of McKean--Vlasov equations are not applicable to \eqref{eq:mvsde-1}--\eqref{eq:mvsde-2}.
The main idea of our  approach is to replace the conditional expectation by the corresponding RKHS approximation \eqref{minv} which has ``nice'' properties (in particular, it is Lipschitz continuous). This would imply strong existence and uniqueness of the new system. Furthermore, we will demonstrate numerically that the solution to the new system is still ``close'' to the solution of \eqref{eq:mvsde-1}--\eqref{eq:mvsde-2} in a certain sense. Thus, we consider the following system:
\begin{align}
	d \wh X_t &= H(t,\wh X_t, Y_t, m^\lambda_{A_1}(\wh X_t;\wh \mu_t) ) d t
	+ F(t,\wh X_t, Y_t, m^\lambda_{A_2}(\wh X_t;\wh \mu_t))\, d W_t^X,\label{newsyst1}\\
	dY_t&=b(t,Y_t)dt+\sigma(t,Y_t)\, dW_t^Y\label{newsyst2}\\
	\wh\mu_t&=\Law(\wh X_t,Y_t).\label{newsyst3}
\end{align}
where $t\in[0,T]$. We need the following assumptions on the kernel $k$ (formulated in a slightly redundant manner for the ease of notation).
\begin{assumption}	\label{AKnew}
	The kernel $k$ is twice continuously differentiable in both variables, $k(x,x)>0$ for  all $x\in\X$, and
	
	\begin{multline*}
		D_{k}^{2}:=\sup_{\substack{(x,y)\in\X\times\X\\1\leq i,j\leq d}}\max\left\{
		|\partial_{x_{i}}\partial_{y_{j}}k^{2}(x,y)|,|\partial_{x_{i}}\partial_{y_{j}%
		}k(x,y)|,|\partial_{x_{i}}k(x,y)|,\right.  \\
		\left.  |\partial_{y_{j}}k(x,y)|,|k(x,y)|\right\} <\infty.
	\end{multline*}
\end{assumption}

Let $\C^1(\X,\R)$ be the space of all functions  $f\colon\X\to\R$ such that
\begin{equation*}
	\|f\|_{\C^1}:=\sup_{x\in\X} |f(x)|+\sup_{\substack{x\in\X\\i=1,\ldots,d}} |\d_{x_i} f(x)|<\infty.
\end{equation*}
Now we are ready to state our main results. Their proofs are given in \cref{sec:proof}.

\begin{theorem}
	\label{prop:exist-reg}
	Suppose that Assumption~\ref{AKnew} is satisfied for the kernel $k$ with $\X = \R^d$ and
	\begin{enumerate}[$(1)$]
		\item $A_i\in\C^1(\R^d,\R)$, $i=1,2$;
		\item there exists a constant $C>0$ such that for any $t\in[0,T]$, ${x,y,x',y'\in \R^d}$, ${z,z'\in\R}$,
		\begin{align*}
			&|H(t,x,y,z)-H(t,x',y',z')|+|F(t,x,y,z)-F(t,x',y',z')|\\
			&\qquad+|b(t,y)-b(t,y')|+|\sigma(t,y)-\sigma(t,y')|
			\\
			&\leq C(|x-x'|+|y-y'|+|z-z'|);
		\end{align*}
		\item for any fixed $x,y,\in \R^d$, $z\in\R$ one has
		\begin{equation*}
			\int_0^T (|H(t,x,y,z)|^2+|F(t,x,y,z)|^2+|b(t,y)|^2+|\sigma(t,y)|^2)\,dt <\infty;
		\end{equation*}
		\item $\E|\wh X_0|^2<\infty$, $\E|Y_0|^2<\infty$.
	\end{enumerate}
	Then for any $\lambda>0$ the system \eqref{newsyst1}, \eqref{newsyst2}, \eqref{newsyst3} with the initial condition $(\wh X_0, Y_0)$ has a unique strong solution.
\end{theorem}

To analyze a numerical scheme solving \eqref{newsyst1}--\eqref{newsyst3}, we consider a particle system
\begin{align}
	d X^{N,n}_t &= H\bigl(t,X^{N,n}_t, Y^{N,n}_t, m^\lambda_{A_1}(X^{N,n}_t;\mu_{t}^{N}) \bigr) d t  \nonumber\\
	&\phantom{=}+ F\bigl(t, X^{N,n}_t, Y^{N,n}_t, m^\lambda_{A_2}(X^{N,n}_t;\mu_{t}^{N})\bigr) d W_t^{X,n},\label{eq:particle-systema}\\
	dY^{N,n}_t&=b(t,Y^{N,n}_t)\,dt+\sigma(t,Y^{N,n}_t)\,dW_t^{Y,n},\label{eq:particle-systemb}\\
	\mu_{t}^{N}&=\frac{1}{N}\sum_{n=1}^{N}\delta_{(X_{t}^{N,n},Y_{t}^{N,n})},\label{eq:particle-systemc}		
\end{align}
where $N\in\N$, $n=1,\ldots,N$, $t\in[0,T]$, and the pairs of $d$-dimensional Brownian motions $(W^{X,n},W^{Y,n})$, $n=1,\ldots,N$, are jointly independent and have the same law as $(W^{X},W^{Y})$.
The following propagation of chaos result holds; it  establishes both weak and strong convergence of $X^{N,n}$.

\begin{theorem}\label{T:MR2}
	Suppose that all the conditions of \cref{prop:exist-reg} are satisfied.
	Suppose that the initial values $(X_0^{N,n},Y_0^{N,n})$ are  jointly independent and have the same law as
	$(\wh X_0,Y_0)$. Moreover, suppose that $\E|\wh X_0|^q<\infty$, $\E|Y_0|^q<\infty$ for some $q>4$. Then
	there exists a constant $C=C(\lambda,T,\E|\wh X_0|^q,\E|Y_0|^q)>0$ such that for any $n=1,$ $\ldots,$ $N,$ $N\in\N,$
	\begin{align}
		\label{eq:prop-chaos}
		\E \left[ \sup_{0\leq t\leq T}|X_{t}^{N,n}-\wh X_{t}^n|^{2} \right]	+\sup_{0\leq t\leq T} \E[\mathbb{W}_{2}(\mu_{t}^{N},\wh \mu_t)^{2}]\le C \epsilon_{N},
	\end{align}
	where the process $\wh X^n$ solves \eqref{newsyst1}--\eqref{newsyst3} with $W^{X,n}$, $W^{Y,n}$ in place of $W^{X}$, $W^{Y}$, respectively,
	and where
	\[
	\epsilon_{N}=\left\{
	\begin{array}
		[c]{r}%
		N^{-1/2}\text{ \ \ if }d=1,\\
		N^{-1/2}\log N\text{ \ \ if }d=2,\\
		N^{-1/d}\text{ \ \ if }d>2.
	\end{array}
	\right.
	\]
\end{theorem}

A crucial step which allowed us to obtain these results is the Lipschitz continuity of $m^\lambda$. The following holds.
\begin{theorem}	\label{lip}
	Assume that the kernel $k$ satisfies Assumption~\ref{AKnew}. Let $A\in\C^1(\X,\R)$. Then for any $x,y\in\X$, $\mu,\nu\in\mathcal{P}_{2}(\X\times\X)$ on has
	\begin{equation*}
		|m^{\lambda}_A(x;\mu)-m^{\lambda}_A(y;\nu)|\le C_1 \W_{1}(\mu,\nu)+C_2|x-y|,
	\end{equation*}
	where
	\begin{equation*}
		C_{1}:=\left(  \frac{D_{k}}{\lambda^{2}}+\frac{1}{\lambda}\right)
		dD_{k}^{2}\| A\|_{\C^1}\text{ \ and \ }C_{2}:=
		\frac{\sqrt{d}}{\lambda}D_{k}^{2}\|A\|_{\C^1}
	\end{equation*}
	may be considered to be (possibly suboptimal) Lipschitz constants with respect to the Wasserstein metric and the Euclidian norm, respectively.
\end{theorem}

This result is interesting for at least two reasons.  First, it shows that $m^\lambda_A$ is Lipschitz continuous in both arguments, provided that the kernel $k$ is smooth enough. That is, the Lipschitz continuity property depends on $\mathcal{H}$ only through the smoothness of the kernel $k$.
Second, this result gives an explicit dependence of the corresponding (possibly suboptimal) Lipschitz constants on $\lambda$ and $k$.

\begin{remark}
	Let us stress that \cref{prop:exist-reg} establishes  the existence and uniqueness of \eqref{eq:mvsde-1}--\eqref{eq:mvsde-2} only for a fixed regularization parameter $\lambda>0$ and can not be used to study the limiting case $\lambda\to 0$. Indeed, it follows from \cref{lip}, that as $\lambda\to0$, the Lipschitz constants of $m^{\lambda}_A$ blow up. Yet, Theorem~2.3 does not imply that for $\lambda\to 0$ the {\em optimal} Lipschitz constants blow up, nor that the solution to \eqref{eq:mvsde-1}--\eqref{eq:mvsde-2} blows up. We will demonstrate numerically  in \cref{sec:app} that for $\lambda\to0$, in the examples  there,   the solution to \eqref{eq:mvsde-1}--\eqref{eq:mvsde-2} does not blow up. On the contrary, it weakly converges to a limit; this suggests that (at least) weak existence of a solution to \eqref{eq:mvsde-1}--\eqref{eq:mvsde-2} may hold. Verifying this theoretically remains however an important open problem.
\end{remark}

\begin{remark}
	\label{rem:general-state-space}
	A natural question is whether \eqref{eq:mvsde-1}--\eqref{eq:mvsde-2} can be formulated for a different state space, that is, for $X,Y$ taking values in $\X$, $\mathcal{Y}$ rather than $\R^d$. Indeed, for equity models, $\X = \mathcal{Y} = \R_+$ is clearly a more natural choice for both the price process and the variance process. Heuristically, the theory will hold for more general $\X$ and $\mathcal{Y}$, provided that those sets are invariant under the dynamics \eqref{eq:mvsde-1}--\eqref{eq:mvsde-2} -- as well as under the regularized dynamics. It is, however, difficult to derive meaningful assumptions guaranteeing this kind of invariance, which prompts us to work with $\R^d$ instead.
\end{remark}

\section{Approximation of conditional expectations }
\label{sec:rkhs}

In this section we study the approximation $m^\lambda_A$ introduced in \eqref{minv} in more detail.
Throughout this section we fix an open set $\X\subset\R^d$, a measure $\nu \in\mathcal{P}_{2}(\X\times\X)$, and  impose the following relatively weak assumptions on the function $A\colon\X\to\R$ and the positive kernel $k\colon\X\times\X$ $\to\R$.
\begin{assumption}
	\label{AKi}
	The function $A$ has sublinear growth, i.e. there exists a constant $C>0$ such that for all $x\in\mathcal{X}$ one has
	$|A(x)|\le C(1+|x|)$.
\end{assumption}
\begin{assumption}
	\label{AKii} The kernel $k(\cdot,\cdot)$ is continuous on $\X\times\X$ and satisfies 
	$$
	0<k(x,x)\le C(1+|x|^{2})
	$$
	for some $C>0$.
\end{assumption}

It is easy to see that Assumption~\ref{AKii} implies for any $x\in\X$
\begin{equation}\label{knorm}
	\|k(x,\cdot)\|_\H^2=\langle k(x,\cdot),k(x,\cdot)\rangle_{\H}=k(x,x)\le C(1+|x|^2).
\end{equation}
Due to Assumption~\ref{AKii} and \cite[Lemma~4.33]{steinwart2008support}, $\H$ is a separable RKHS and one has for any $f\in\mathcal{H}$,
$x\in\mathcal{X}$,%
\begin{equation}\label{mb}
	| f(x)| =|\langle k(x,\cdot),f\rangle_{\H}|\le \| k(x,\cdot)\|_{\H} \|
	f\|_{\H}\le 
	{C}(1+|x|)\|
	f\|_{\H},
\end{equation}
where we also used \eqref{knorm}. Hence, every
$f\in\mathcal{H}$ has sublinear growth and, as a consequence, for any fixed
$\nu\in\mathcal{P}_{2}(\X\times\X)$, the objective functional in \eqref{minv} is finite.
It is also easy to see that \eqref{mb} and \eqref{knorm} imply that for any $x,y\in\X$
\begin{equation}\label{kzxbound}
	|k(x,y)|\le {C}(1+|x|)\|k(\cdot,y)\|_\H\le C(1+|x|) (1+|y|).
\end{equation}
Therefore, the Bochner integrals
\begin{equation*}
	c_A^{\nu}:=\int_{\X\times\X} k(\cdot,x)A(y)\nu(dx,dy),\quad \text{and}\quad
	\mathcal{C}^{\nu}f  :=\int_{\X\times\X} k(\cdot,x)f(x)\nu   
	(dx,dy).
\end{equation*}
are well defined functions in $\H$ for every $f\in\H$.
Moreover, it is clear that the operator ${\mathcal{C}^{\nu}\colon\H \to\H}$ is
symmetric and positive semidefinite since%
\begin{equation*}
	\left\langle g,\mathcal{C}^{\nu}f\right\rangle_{\mathcal{H}}   =\int_{\X}\left\langle
	g,k(\cdot,x)\right\rangle f(x)\nu(dx,\mathcal{X})
	=\int_{\X} g(x)f(x)\nu(dx,\mathcal{X}).
\end{equation*}
Thus, by the Hellinger--Toeplitz theorem (see, e.g., \cite[Section~III.5]{ReedSimon}), $\mathcal{C}^{\nu}$ is a boun\-ded self-adjoint linear operator on $\H$.
As a consequence,
for any $\lambda\ge0$, the operator $\mathcal{C}^{\nu}+\lambda I_{\mathcal{H}}$ is a
bounded self-adjoint operator on $\mathcal{H}$ with spectrum contained in the
interval $[ \lambda,\|\mathcal{C}^{\nu}\|
+\lambda]$.
Hence, if $\lambda>0$, then $(\mathcal{C}^{\nu}+\lambda I_{\mathcal{H}})^{-1}$
exists and is a bounded self-adjoint operator on $\mathcal{H}$ with norm
\begin{equation*}
	\|(\mathcal{C}^{\nu}+\lambda I_{\mathcal{H}})^{-1}\|_{\H}\le \lambda^{-1}.
\end{equation*}
We
are now ready to state the following useful representation for the solution to (\ref{minv}).
\begin{proposition}\label{prop: apcon}
	Under Assumptions~\ref{AKi}, \ref{AKii}, for any fixed $\nu\in\mathcal{P}_{2}(\X\times\X)$ and $\lambda>0,$ the
	solution to (\ref{minv}) can be represented as
	\begin{equation}\label{identm}
		m^{\lambda}_A(\cdot;\nu)=(\mathcal{C}^{\nu}+\lambda I_{\mathcal{H}})^{-1}c^{\nu}_A.
	\end{equation}
\end{proposition}
This representation
may be seen as an infinite sample version of  the usual solution representation for a ridge regression problem based on finite samples.
We thus consider it as not essentially new, but, in order to keep our paper as self contained as possible we present a proof of it in \cref{sec:proof}.
\cref{prop: apcon} allows us to prove Lipschitz continuity of $m^\lambda_A$, that is \cref{lip}.
\par
Let us now proceed with investigating when the function $m^{\lambda}_A=m^{\lambda}_A(\cdot;\nu)$ is a
``good'' approximation to the true conditional expectation
\begin{equation}\label{mnuu}
	m_A=m_A(x;\nu):=\mathsf{E}_{(X,Y)\sim\nu}[  A(Y)|X=x]
\end{equation}
for small enough  $\lambda>0$.
Consider the Hilbert space $\mathcal{L}_2^\nu:=L_2(\X, \nu(dx,\mathcal{X}))$ with $\nu
(U,\mathcal{X}) \coloneqq \nu(U \times \mathcal{X}) >0.$
For $f\in\L_2^\nu$ put
\begin{equation} \label{to}
	T^\nu f:=\int_\X k(\cdot,x)f(x)\nu(dx,\X).
\end{equation}
Recalling \eqref{kzxbound}, it is easy to see that $T^\nu$ is a linear operator $\L_2^\nu\to\L_2^\nu$. Note that
that  $\mathcal{H}\subset \mathcal{L}_{2}^\nu$ due to \eqref{mb}; thus, $\mathcal{C}^\nu$ is the restriction of $T^\nu$ to $\H$.
Further, since 
$$|k(x,y)| \leq \sqrt{k(x,x)} \sqrt{k(y,y)},
$$
the kernel $k$ is Hilbert-Schmidt
on $\mathcal{L}_{2}(\mathcal{X\times X},\nu
(dx,\mathcal{X}) \otimes \nu(dy,\mathcal{X}))$, i.e.
\begin{equation*}
	\int k^{2}(x,y)\nu(dx,\mathcal{X})\nu(dy,\mathcal{X})<\infty,
\end{equation*}
due to Assumption~\ref{AKii}.
As a consequence of the standard results from functional analysis, one then has (see, for example, \cite[Section~VI]{ReedSimon}):
\begin{enumerate}[(i)]
	\item the operator $T^\nu$ is self-adjoint and compact;
	\item there exists an orthonormal system $\left(  a_{n}\right)  _{n\in\mathbb{N}}$ in $\mathcal{L}^\nu_{2}$ of 
	eigenfunctions corresponding to nonnegative eigenvalues $\sigma_{n}$ of $T^\nu$ and
	$\sigma_{1}\ge\sigma_{2}\ge\sigma_{3}\ge\ldots$;
	\item If $J:=\{n\in\N:$ $\sigma_{n}>0\},$ one has
	\begin{equation}\label{compop}
		T^\nu f=\sum_{n\in J}\sigma_{n}\left\langle f,a_{n}\right\rangle _{\mathcal{L}^\nu
			_{2}}a_{n},\text{ \ \ }f\in\mathcal{L}_{2}^\nu%
	\end{equation}
	with $\lim_{n\rightarrow\infty}\sigma_{n}=0$ if $J=\N.$
\end{enumerate}

A generalization of Mercer's theorem to unbounded domains  \cite{Sun2005}
implies the following statement.
\begin{proposition}
	\label{mercer0}Let $k$ be a kernel satisfying Assumption~\ref{AKii} and 
	assume that  $\nu(\cdot,\mathcal{X})$  is a nondegenerate Borel measure. That is, for every open set $U\subset\mathcal{X}$ one has $\nu(U,\mathcal{X})>0$. Then one may take the eigenfunctions $a_n$ in (\ref{compop}) to be continuous and   $k$ has a series representation%
	\begin{equation*}
		k(x,y)=\sum_{n\in J}\sigma_{n}a_{n}(x)a_{n}(y),\quad x,y\in\X \label{invar}
	\end{equation*}
	with uniform convergence on compact sets.
	Moreover, $(\widetilde{a}_{n})_{n\in
		J}$ with $\widetilde{a}_{n}:=\sqrt{\sigma_{n}}a_{n}$ is an orthonormal basis of
	$\mathcal{H}$ 
	and the scalar product in $\mathcal{H}$ takes the form%
	\begin{equation}
		\langle f,g\rangle _{\H}=\sum_{n\in J}\frac{\langle f,a_{n}\rangle_{\mathcal{L}_2^\nu}
			\langle g,a_{n}\rangle_{\mathcal{L}_2^\nu}}{\sigma_n}\text{ \ \ for \ \ }f,g\in\mathcal{H}.%
		\label{scpr}
	\end{equation}
\end{proposition}

Now we are ready to present the main result of this section, which quantifies the convergence properties of
$m^\lambda_A(\cdot,\nu)$ as $\lambda\to0$ for a
fixed measure $\nu$. Recall the notation \eqref{mnuu}.
Let $P_{\overline{\mathcal{H}}}$ denote the orthogonal projection in $\mathcal{L}%
_{2}^{\nu}$ onto $\overline{\mathcal{H}},$ i.e. the closure of $\mathcal{H}$ in $\mathcal{L}%
_{2}^{\nu}.$ Hence for any $f\in\mathcal{L}%
_{2}^{\nu}$,
\begin{equation}
	P_{\overline{\mathcal{H}}}f=\sum_{n\in J}\left\langle f,a_{n}\right\rangle_{\mathcal{L}_2^\nu}
	a_{n}\qquad\text{and}\qquad\left\langle P_{\overline{\mathcal{H}}}%
	f,a_{m}\right\rangle_{\mathcal{L}_2^\nu} =\left\langle f,a_{m}\right\rangle_{\mathcal{L}_2^\nu} ,\quad m\in
	J,\label{QRef}%
\end{equation}
since $(a_{n})_{n\in J}$ is an orthonormal system in $\mathcal{L}_{2}^{\nu}$.

\begin{theorem}
	\label{L2c}  Assume that the kernel $k$ satisfies
	Assumption~\ref{AKii}, $\nu(\cdot,\mathcal{X})$  is a nondegenerate Borel measure, and that $m_{A}(\cdot;\nu)\in\mathcal{L}_{2}^\nu$ (for instance, because $A$ is bounded measurable). Then for any $\lambda>0$
	\begin{equation}
		\left\|P_{\overline{\mathcal{H}}}m_{A}(\cdot;\nu)-m_{A}^{\lambda}(\cdot;\nu)\right\|_{\mathcal{L}_2^\nu}^{2}=\sum_{n\in J}\frac{\lambda^{2}}{\left(  \sigma
			_{n}+\lambda\right)  ^{2}}\left\langle m_{A}(\cdot; \nu),a_{n}\right\rangle _{\mathcal{L}_2^\nu}^{2}. \label{conver}%
	\end{equation}
	In particular, $\left\| P_{\overline{\mathcal{H}}}m_{A}(\cdot;\nu)-m_{A}^{\lambda}(\cdot;\nu)\right\|_{\mathcal{L}_2^\nu}\rightarrow0$ as $\lambda\searrow0$.
	If, moreover, $P_{\overline{\mathcal{H}}}m_{A}(\cdot;\nu)\in\mathcal{H}$ one has%
	\begin{equation}
		\left\| P_{\overline{\mathcal{H}}}m_{A}(\cdot;\nu)-m(\cdot;\nu)_{A}^{\lambda}(\cdot;\nu)\right\|
		_{\mathcal{H}}^{2}=\sum_{n\in J}\frac{\lambda^{2}}{\left(  \sigma_{n}%
			+\lambda\right)  ^{2}\sigma_{n}}\left\langle m_{A}(\cdot;\nu),a_{n}\right\rangle
		_{\mathcal{L}_2^\nu}^{2},\label{conver1}%
	\end{equation}
	and thus $\left\| P_{\overline{\mathcal{H}}}m_{A}(\cdot;\nu)-m_{A}^{\lambda
	}(\cdot;\nu)\right\|_{\mathcal{H}}\rightarrow0$ for $\lambda\searrow0.$
\end{theorem}

\cref{L2c} establishes convergence of $m_{A}^{\lambda}(\cdot;\nu)$ as $\lambda\to0$ though without a rate. Its proof is placed in \cref{sec:proof}. Additional assumptions are needed to guarantee a certain convergence rate. This is done in the following corollary.
\begin{corollary}
	\label{exc} Suppose that the conditions of
	\cref{L2c} are satisfied, and that moreover for some $\theta\in(0,1]$,
	\begin{equation}
		\sum_{n\in J}\sigma_{n}^{-\theta}\left\langle 
		m_A(\cdot;\nu),a_{n}\right\rangle _{\mathcal{L}_2^\nu}^{2}<\infty. \label{condH}%
	\end{equation}
	Then
	\begin{align}
		\left\|  P_{\overline{\mathcal{H}}}m_A(\cdot;\nu)  -m^{\lambda}_A(\cdot;\nu)\right\| _{\mathcal{L}_2^\nu}^2
		\leq\left(  1-\frac{\theta}{2}\right)  ^{2}\left(  \frac{\lambda\theta
		}{2-\theta}\right)  ^{\theta}\sum_{n\in J}\sigma_{n}^{-\theta}\left\langle
		m_{A}(\cdot;\nu),a_{n}\right\rangle _{\mathcal{L}_2^\nu}^{2}.\label{corth}
	\end{align}
	In particular, if $\theta=1$ then $ P_{\overline{\mathcal{H}}}m_A  \in\mathcal{H}$,  and we get
	\begin{equation}\label{res2}
		\left\|  P_{\overline{\mathcal{H}}}m_A(\cdot;\nu)  -m^{\lambda}_A(\cdot;\nu)\right\Vert _{\mathcal{L}_2^\nu}
		\leq \frac{\sqrt{\lambda}}2\left\| P_{\overline{\mathcal{H}}}m_A(\cdot;\nu)  \right\|_{\mathcal{H}}.
	\end{equation}
\end{corollary}
\begin{proof}
	Inequality \eqref{corth} follows from \eqref{conver}, \eqref{condH}, and the fact that the maximum of the function $x\mapsto\lambda^{2}x^{\theta}/(x+\lambda)^2$, $x>0$, is equal to
	$$
	{(1-\theta/2)  ^2(\lambda\theta/(2-\theta))  ^{\theta}}.
	$$
	Inequality \eqref{res2} follows from \eqref{scpr}, \eqref{QRef} and \eqref{corth}.
\end{proof}

\begin{remark} If operator $T^{\nu}$ defined in (\ref{to}) is injective, that is, $T^{\nu}f = 0$ for $f\in\mathcal{L}_{2}^{\nu}$ implies $f = 0,$ $\nu$-a.s., then
	$P_{\overline{\mathcal{H}}}=I_{\mathcal{L}_{2}^{\nu}}$. In this case, $J=\N$ and
	\cref{L2c} and \cref{exc} quantify the convergence to the true conditional
	expectation. A sufficient condition for $T^{\nu}$ to be injective is that the
	kernel $k$ is \emph{integrally strictly positive definite (ispd),} in the
	sense that%
	\[
	\int_{\mathcal{X}\times\mathcal{X}}k(x,y)\mu(dx)\mu(dy)>0
	\]
	for all non-zero signed Borel measures $\mu$ defined on $\mathcal{X}$. Indeed,
	for any $f\in\mathcal{L}_{2}^{\nu}$ we may define a signed Borel measure
	$\mu_{f}(A):=\int_{A}f(x)\nu(dx,\mathcal{X}),$ $A\in\mathcal{B}(\mathcal{X}),$
	which is finite since $\left\vert \mu_{f}(A)\right\vert \leq\int f^{2}%
	(x)\nu(dx,\mathcal{X})<\infty$. Hence, if $k$ is an ispd kernel then $T^{\nu
	}f = 0$ implies
	\begin{align*}
		0  & =\left\langle T^{\nu}f,f\right\rangle _{\mathcal{L}_{2}^{\nu}}%
		=\int_{\mathcal{X}\times\mathcal{X}}k(y,x)f(x)f(y)\nu(dx,\mathcal{X}%
		)\nu(dy,\mathcal{X})\\
		& =\int_{\mathcal{X}\times\mathcal{X}}k(y,x)\mu_{f}(dx)\mu_{f}(dy)
	\end{align*}
	which in turn implies $\mu_{f}=0,$ i.e. $f=0,$ $\nu$-a.s. Further it should be
	noted that any ispd kernel is strictly positive definite in the usual sense,
	but the converse is not true. Examples of ispd kernels are Gaussian kernels,
	Laplace kernels, and many more. For details on ispd kernels we refer to
	\cite{SGFSL}.
\end{remark}

Thus, in this section we have shown that, under certain conditions, $m^{\lambda}_A(\cdot,\nu)$ may converge
at least in ${\mathcal{L}_2^\nu}$-sense to the true conditional expectation $m_A(\cdot,\nu)$ as ${\lambda\to0}$. This makes the heuristic discussion around \eqref{mla0} and \eqref{mla} in \cref{sec:introduction} more rigorous.

\begin{remark}
	Note that the measure $\wh\mu_t$ in the solution of \eqref{newsyst1}--\eqref{newsyst3} depends on $\lambda$, so in fact
	$\wh\mu_t=\wh\mu_t^\lambda$. Therefore, even when $m^{\lambda}_A(\cdot,\nu)\to m_A(\cdot,\nu)$ for fixed $\nu$ and ${\lambda\downarrow0}$, the question whether $m^{\lambda}_{A_i}(\cdot,\wh \mu_t^\lambda)$ converges in some sense is still not answered. We believe that this question
	is intimately linked to the problem of existence of a solution to \eqref{eq:mvsde-1}--\eqref{eq:mvsde-2}. As already explained, this is an unsolved open problem
	and therefore considered out of our scope. However, loosely speaking, assuming that the latter system has indeed a solution (in some sense) with solution measure $\mu_t$ say, it is natural to expect that for a suitable  ``rich enough'' RKHS,
	$m^{\lambda}_{A_i}(\cdot,\mu_t)\to m_{A_i}(\cdot, \mu_t)$ (the true conditional expectation) as $\lambda\searrow0$.
\end{remark}

\section{Numerical algorithm}
\label{sec:alg}

Let us now describe in detail our numerical algorithm to construct solutions to
\eqref{eq:X-dynamics}.  We begin by discussing an efficient way of calculating $m^\lambda_A$.

\subsection{Estimation of the conditional expectation}\label{subsec:est}
Let us recall that in order to solve the particle system \eqref{eq:particle-systema}, \eqref{eq:particle-systemb}, \eqref{eq:particle-systemc}
we need to compute
\begin{equation}\label{minv_a}
	m^\lambda_A(\cdot; \mu^N_t) = \argmin_{f \in \mathcal{H}} \Bigl\{ \frac{1}{N} \sum_{n=1}^N |A(Y_t^{N,n}) - f(X_t^{N,n})|^2 + \lambda \norm{f}_{\mathcal{H}}^2 \Bigr\}.
\end{equation}
for $t$ belonging to a certain partition of $[0,T]$ and fixed large $N\in\N$; here $A=A_1$ or $A=A_2$.
It follows from the representer theorem for RKHS \cite[Theorem~1]{scholkopf2001generalized} that $m^\lambda_A$ has the following representation:
\begin{equation}
	\label{eq:m-lambda-N-representer}
	m^\lambda_A(\cdot; \mu^N_t) = \sum_{i=1}^N \alpha_i k(X_t^{N,i}, \cdot),
\end{equation}
for some $\alpha = (\alpha_1, \ldots, \alpha_N)^T \in \mathbb{R}^N$. Note that the optimal $\alpha$ can be calculated explicitly by plugging the representation~\eqref{eq:m-lambda-N-representer} into the above minimization problem in place of $f$ and minimizing over $\alpha$. However, computing the optimal $\alpha$ directly takes  $O(N^3)$ operations, which is prohibitively expensive keeping in mind that the number of particles $N$ is going to be very large. Furthermore, even  evaluating \eqref{eq:m-lambda-N-representer} at $X_t^{N,n}$, $n=1,\ldots, N$, for a given $\alpha\in\R^N$ is rather expensive, it requires $O(N^2)$ operations, and thus is impossible to implement.

To develop an efficient algorithm, let us note that many particles $X_t^{N,i}$ --- and, as a consequence, the implied basis functions $k(X_t^{N,i}, \cdot)$ --- will be close to each other. Therefore, we can considerably reduce the computational cost by only using ${L\ll N}$ rather than $N$ basis functions as suggested in~\eqref{eq:m-lambda-N-representer}.
More precisely, we choose $Z^1, \ldots, Z^L$ among $X_t^{N,1}, \ldots, X_t^{N,N}$ -- e.g., by random choice or taking every $\frac{N}{L}$th point among the ordered sequence  $X_t^{N,(1)}, \ldots, X_t^{N,(N)}$ in case when $X$ is one-dimensional -- and approximate
\begin{equation*}
	\sum_{i=1}^N \alpha_i k(X_t^{N,i}, \cdot) \approx \sum_{j=1}^L \beta_j k(Z^j, \cdot),
\end{equation*}
where  $\beta = (\beta_1, \ldots, \beta_L)^T \in \mathbb{R}^L$.
It is easy to see that
\begin{align*}
	\Bigl\|\sum_{j=1}^L \beta_j k(Z^j,\cdot)\Bigr\|_{\H}^2&=
	\Bigl\langle\sum_{j=1}^L \beta_j k(Z^j,\cdot),\sum_{j=1}^L \beta_j k(Z^j,\cdot)\Bigr\rangle_{\H}\\
	&=\sum_{j,k=1}^L \beta_j\beta_k \langle k(Z^j,\cdot),k(Z^k,\cdot)\rangle_{\H}\\
	&=\sum_{j,k=1}^L \beta_j\beta_k  k(Z^j,Z^k)=\beta^\top R \beta,
\end{align*}
where $R:= (k(Z^j,Z^k))_{j,k=1,\ldots,L}$ is an $L\times L$ matrix. Thus, recalling~\eqref{minv_a}, we see that we have to solve
\begin{equation*}
	\argmin_{\beta\in\R^L} [\frac1N(G-K\beta)^\top(G-K\beta)+\lambda\beta^\top R \beta],
\end{equation*}
where $G:=(A(Y_t^{N,n}))_{n=1,\ldots,N}$, $K \coloneqq (k(Z^j,X_t^{N,n}))_{n=1,\ldots, N, j=1,\ldots,L}$ is an $N\times L$ matrix. Differentiating with respect to $\beta$, we get that the optimal value $\widehat{\beta}=\widehat{\beta}((X_t^{N}),(Y_t^{N}))$ satisfies
\begin{equation}\label{maineq}
	(K^\top K+N\lambda R)\widehat{\beta}=K^\top G,
\end{equation}
and we approximate expectation as
\begin{equation}\label{phidef}
	m^\lambda_A(x; \mu^N_t) \approx\sum_{j=1}^L \widehat{\beta}_j k(Z^j,x) \eqqcolon \widehat{m}_A^\lambda(x; \mu^N_t).
\end{equation}
\begin{remark}
	\label{rem:grid-method-GHL}
	The method of choosing basis points $Z^1, \ldots, Z^L$ can be seen as a systematic and adaptive approach of choosing basis functions $k(Z^j, \cdot)$, $j=1, \ldots, L$, in  global regression method. We note that the technique of evaluating the conditional expectation only in points on a grid $G_{f,t}$ coupled with spline-type interpolation between grid points suggested in \cite{G-HL} is motivated by similar concerns regarding explosion of computational cost.
\end{remark}
\begin{remark}\label{R:nop}
	Let us see how many operations we need to calculate $\widehat{\beta}$, taking into account that $L\ll N$. We need $O(NL)$ to calculate $K$, $O(L^2)$ to calculate $R$, $O(N L^2)$ to calculate $K^\top K$ (this is the bottleneck); $O(L^3)$ to invert   $K^\top K+N\lambda R$ and $O(NL)$ to calculate $K^\top G$ and solve \eqref{maineq}. Thus, in total we would need $O(N L^2)$ operations.
\end{remark}

\subsection{Solving the regularized McKean--Vlasov equation}

With the function $\widehat{m}_A^\lambda$ in hand, we now consider the Euler scheme for the particle system \eqref{eq:particle-systema}--\eqref{eq:particle-systemc}. We fix a time interval $[0,T]$, the number of time steps $M$, and, for simplicity, we consider a uniform time increment $\delta:=T/M$.  Let $\Delta W_i^{X,n}$ and $\Delta W_i^{Y,n}$ denote independent copies of $W^X_{(i+1) \delta} - W^X_{i\delta}$ and $W^Y_{(i+1) \delta} - W^Y_{i \delta}$, respectively, $n=1, \ldots, N$, $i=1,\ldots, M$. Note that for stochastic volatility models, the Brownian motions driving the stock price and the variance process are usually correlated.
We now define $\wt X_0^n=X_0^n$, $\wt Y_0^n=Y_0^n$, and for $i=0,\ldots, M-1,$
\begin{align}
	\widetilde{X}^n_{i+1} &= \widetilde{X}^n_i + H\left(i\delta, \widetilde{X}^n_i, \widetilde{Y}^n_i, \widehat{m}^\lambda_{A_1}(\widetilde{X}^n_i; \widetilde{\mu}^N_i) \right) \delta\nn \\	&\quad +F\left(i\delta, \widetilde{X}^n_i, \widetilde{Y}^n_i, \widehat{m}^\lambda_{A_2}(\widetilde{X}^n_i; \widetilde{\mu}^N_i) \right) \Delta W_i^{X,n}	\label{eq:Euler-mvsdea}\\
	\widetilde{Y}^n_{i+1} &= \widetilde{Y}^n_i + b(i\delta,\widetilde{Y}^n_i) \delta + \sigma(i\delta,\widetilde{Y}^n_i) \Delta W^{Y,n}_i,	\label{eq:Euler-mvsdeb}
\end{align}
where $\wt \mu_{i}^{N}=\frac{1}{N}\sum_{n=1}^{N}\delta_{(\wt X_{i}^{N,n},\wt Y_{i}^{N,n})}$. 
Thus, at each discretization time step of \eqref{eq:Euler-mvsdea}--\eqref{eq:Euler-mvsdeb} we need to compute approximations of the conditional expectations $\widehat{m}^\lambda_{A_r}(\widetilde{X}^n_i; \widetilde{\mu}^N_i)$, ${r=1,2}$. This is done using the algorithm discussed in \cref{subsec:est}, and takes $O(NL^2)$ operations, see \cref{R:nop}. Thus the total number of operations needed to implement \eqref{eq:Euler-mvsdea}--\eqref{eq:Euler-mvsdeb} is $O(MNL^2)$.

\section{Numerical examples and applications to local stochastic volatility models}
\label{sec:app}

As a main application of the regularization approach presented above, we consider the problem of calibration of stochastic volatility models to market data.
Fix time period $T>0$. To simplify the calculations, we suppose that the interest rate $r=0$. Let $C(t,K)$, $t\in[0,T]$, $K\ge0$, be the  price at time $0$ of a European  call option on a non-dividend paying stock with strike $K$ and maturity $t$. We assume that the market prices $(C(t,K))_{t\in[0,T], K\ge0}$ are given and satisfy the following conditions: $C$ is continuous and increasing in $t$, twice continuously differentiable in $x$, $\partial_{xx}C(t,x)>0$, $C(t,x)\to0$ as $x\to\infty$ for any $t\ge0$, $C(t,0)=const$. It is known \cite[Theorem~1.3 and Section~2.1]{Low08},  \cite{Dup} that under these conditions there exists a diffusion process $(S_t)_{t\in[0,T]}$ which is able to perfectly replicate the given call option
prices, that is $\E (S_t-K)^+=C(t,K)$. Furthermore, $S$ solves the following stochastic differential equation
\begin{equation}\label{processS}
	dS_{t}=\sigma_{\text{Dup}}(t,S_{t})S_{t}\,dW_{t},\quad t\in[0,T],
\end{equation}
where $W$ is a Brownian motion and  $\sigma_{\text{Dup}}$ is the Dupire local volatility given by
\begin{equation}
	\sigma_{\text{Dup}}^{2}(t,x):=\frac{2\partial_{t}C(t,x)}{x^{2}\partial
		_{xx}C(t,x)},\quad x>0,\,t\in[0,T].\label{duplv}%
\end{equation}
We study Local Stochastic Volatility (LSV) models. That is, we assume that the
stock price $X$ follows the dynamics
\begin{equation}
	dX_{t}=\sqrt{Y_{t}}\sigma_{\text{LV}}(t,X_{t})X_{t}\, dW_{t}^{X},\quad t\in[0,T],\label{LSV}%
\end{equation}
where $W^X$ is a Brownian motion and $(Y_t)_{t \in [0,T]}$ is a strictly positive
variance process, both being adapted to some filtration $\left(
\mathcal{F}_{t}\right)_{t\geq0}$. If the function $\sigma_{\text{LV}}$ is given by
\begin{equation*}
	\sigma_{\text{LV}}^{2}(t,x) \coloneqq \frac{\sigma_{\text{Dup}}^{2}(t,x)}{\E\left[
		\left.  Y_{t}\right\vert X_{t}=x\right]  },\quad x>0,\,t\in[0,T],
\end{equation*}
and $\int_0^T \E \left[ Y_{t}\sigma_{\text{LV}}(t,X_{t})^2 X_{t}^2 \right] \,dt<\infty$, then marginal distributions of $X_t$ coincide with marginal distributions of $S_t$ (\cite[Theorem~4.6]{Gyo}, \cite[Corollary~3.7]{BS2013}). Thus,
\begin{equation}
	C(T,K)=\E(X_{T}-K)_{+},\quad T,K>0.\label{priceidentity}%
\end{equation}
In particular, the choice $Y\equiv1$ recovers the  local volatility model.
In case where $Y$ is a diffusion process
\begin{equation}\label{eqY}
	dY_t= b(t,Y_t)dt+\sigma(t,Y_t) d W_t^Y,
\end{equation}
where $W^Y$ is a Brownian motion possibly correlated with $W^X$, we see that the model \eqref{LSV}-\eqref{eqY} is a special case of the general McKean-Vlasov equation  \eqref{eq:mvsde-1}--\eqref{eq:mvsde-2}.
To solve \eqref{LSV}-\eqref{eqY}, we implement the algorithm described in \cref{sec:alg}: see \eqref{eq:Euler-mvsdea}--\eqref{eq:Euler-mvsdeb} together with~\eqref{phidef}. We validate our results, by doing two different checks.
First, we verify that 1-dimensional marginal distributions $(\wt X_M)$ are close to the correct marginal distribution $\Law(X_T)=\Law(S_T)$. To do it we compare the call option prices obtained by the algorithm (that is $N^{-1} \sum_{n=1}^N(\wt X_M^n-K)^+$)  with
the given prices $C(T,K)$ for various $T>0$ and $K>0$. If the algorithm is correct and if $\wt \mu^N_M\approx \Law(X_T,Y_T)$, then, according to \eqref{priceidentity}, one must have
\begin{equation}\label{priceidentity2}
	C(T,K)\approx N^{-1} \sum_{n=1}^N(\wt X_M^n-K)^+=:\wt C(T,K).
\end{equation}
On the other hand, if the algorithm is not correct and  $\Law(X_T,Y_T)$ is very different from $\wt \mu^N_M$, then \eqref{priceidentity2} will not hold.

Second, we also control the multivariate distribution of  $(\wt X_i)_{i=0,...,M}$. Recall that for any $t\in[0,T]$ we have $\Law(X_t)=\Law(S_t)$.
We want to make sure that the dynamics of  process $\wt X$ is different from the dynamics of the local volatility process $S$. 
As a test case, we compare option values on the  quadratic variation of the  logarithm of the price. More precisely, for each $K>0$ we compare European options on quadratic variation:
\begin{align*}
	&QV_S(K):= \frac{1}{N} \sum_{n=1}^N \Bigl(\sum_{i=0}^M (\log(S_{(i+1)T/M}^n)-\log(S_{iT/M}^n))^2 - K\Bigr)^+\\
	&QV_{\wt X}(K):= \frac{1}{N} \sum_{n=1}^N \Bigl(\sum_{i=0}^M (\log(\wt X^n_{i+1})-\log(\wt X^n_{i}))^2-K\Bigr)^+
\end{align*}
and verify that these two curves are different. Here, $(S^n)_{i=1,..,N}$ is an Euler approximation of \eqref{processS}. We also check that prices of European options on quadratic variation stabilize as $N \to \infty$.

We will consider two different ways to generate market prices $C(T,K)$. First, we assume that the stock follows the Black--Scholes (BS) model, that is, we assume $\sigma_{\text{Dup}}\equiv const$ for $const=0.3$ and $S_0=1$.
Second, we consider a stochastic volatility model  for the market, that is, we set $C(T,K):=\E\left[ (\overline{S}_T-K)^+\right]$, where $\overline{S}_t,$ $t>0,$  follows the Heston model
\begin{align}
	d\overline{S}_{t}&={\sqrt {v _{t}}}\overline{S}_{t}\,dW_{t},\label{HMSa}\\
	dv_{t}&=\kappa (\theta -v _{t})\,dt+ \xi {\sqrt {v _{t}}}\,dB_t,\label{HMSb}
\end{align}
with the following parameters: $\kappa = 2.19$, $\theta = 0.17023$, $\xi = 1.04$,  and correlation $\rho = -0.83$ between the driving Brownian motions $W$ and $B$, with initial values $\overline{S}_0=1$,  $v_0=0.0045$, cf.~similar parameter choices in \cite[Table~1]{LMP}. We compute option prices based on \eqref{HMSa}--\eqref{HMSb} with the COS method, see \cite{fang2009novel}. We then calculate $\sigma_{\text{Dup}}$ from $C(T,K)$ using \eqref{duplv}.

As our baseline stochastic volatility model for $Y$, we choose a capped-from-below Heston-type model, but with different parameters than the data-generating Heston model. Specifically, we set $b(t,x)=\lambda(\mu-x)$ and $\sigma(t,x)=\eta\sqrt{x}$ in \eqref{eqY}, where $Y_0=0.0144$, $\lambda=1$, $\mu=0.0144$, and $\eta=0.5751$. We cap the solution of \eqref{eqY} from below at the level $\varepsilon_{CIR}=10^{-3}$ to avoid singularity at $0$. Numerical experiments have shown that such capping is necessary. We assume that the correlation between $W^X$ and $W^Y$ is very strong and equals $-0.9$, which makes calibration more difficult. Since the variance process has different parameters compared to the price-generating stochastic volatility model, a non-trivial local volatility function is required to match the implied volatility. Hence, even though the generating model is of the same \emph{class}, the calibration problem is still non-trivial, and involves a singular MKV SDE.

We took $\H$ to be RKHS associated with the Gaussian kernel $k$ with variance $0.1 $.
We fix the number of time steps $M=500$, $\lambda=10^{-9}$, $L=100$. At each time step of the Euler scheme we choose $(Z^j)_{j=1,\ldots, L}$ by the following rule:
\begin{equation}\label{choice}
	\text{$Z_j$ is the $j\cdot 100/(L+1)$ percentile of the sequence  $\{ \widetilde{ X}^n_m\}_{n=1, \ldots, N}$,}
	\end{equation}
	an approach comparable to the choice of the evaluation grid $G_{f,t}$ suggested in \cite{G-HL}.
	
	\Cref{f:2} compares the theoretical and the calculated prices (in terms of implied volatilities)  in the Black-Scholes (a) and Heston (b-d) settings for various strikes and maturities.
	That is, we first calculate $C(T,K)$ using the Black-Scholes model (``Black-Scholes setting'') or  \eqref{HMSa}--\eqref{HMSb} (``Heston setting''); then we calculate $\sigma^2_{\text{Dup}}$ by \eqref{duplv}; then we calculate $\wt X^n_M$, $n=1,\ldots,N,$ using the algorithm \eqref{eq:Euler-mvsdea}--\eqref{eq:Euler-mvsdeb} with $H\equiv0$, $A_2(x)=x$, and
	$$
	F(t,x,y,z):=x \sigma_{\text{Dup}}(t,x)\frac{\sqrt y}{\sqrt {z\vee\eps} },
	$$
	where $\eps=10^{-3}$, then  we calculate $\wt C(T,K)$ using \eqref{priceidentity2}; finally we transform the prices $C(T,K)$ and
	$\wt C(T,K)$  to the implied volatilities. 
	We would like to note that this additional capping of the function $F$ is less critical than the capping of the baseline process $Y$.
	
	We plot at \Cref{f:2} implied volatilities for a wide range of strikes and maturities. More precisely, we consider all strikes $K$ such that ${\P(S_T<K)}\in[0.05,0.95]$ --- this corresponds to all but very far in--the--money and out--of--the--money options.
	One can see from \Cref{f:2} that already for $N=10^3$ trajectories, identity~\eqref{priceidentity2} holds up to a small error for all the considered strikes and maturities.  This error further diminishes as the number of trajectories increases. At $N=10^5$ the true implied volatility curve and the one calculated from our approximation model become almost indistinguishable.

	\begin{figure}
\centering
\subfloat[]{\includegraphics[width=0.49\textwidth]{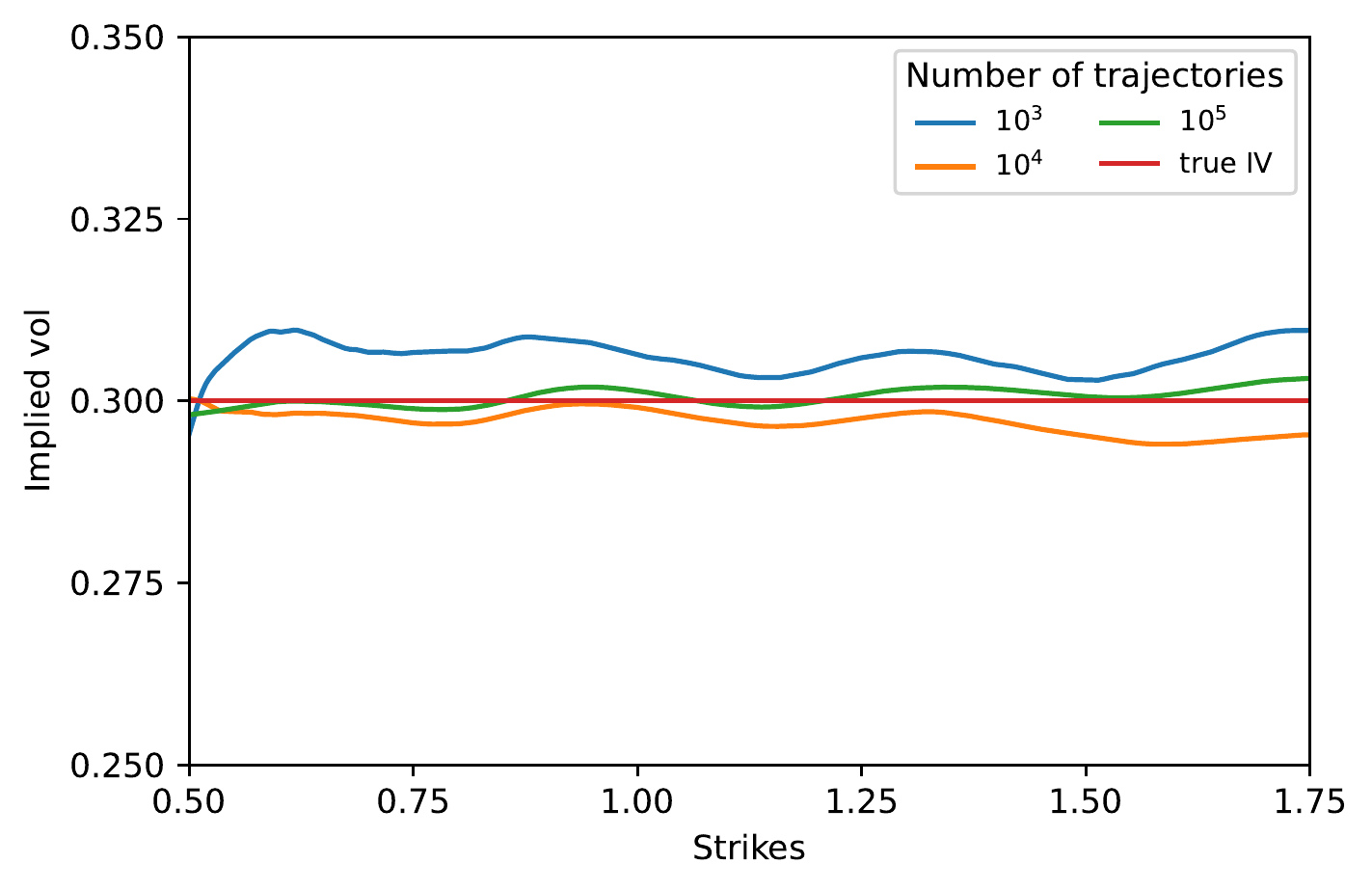}}
\hfill
\subfloat[]{\includegraphics[width=0.49\textwidth]{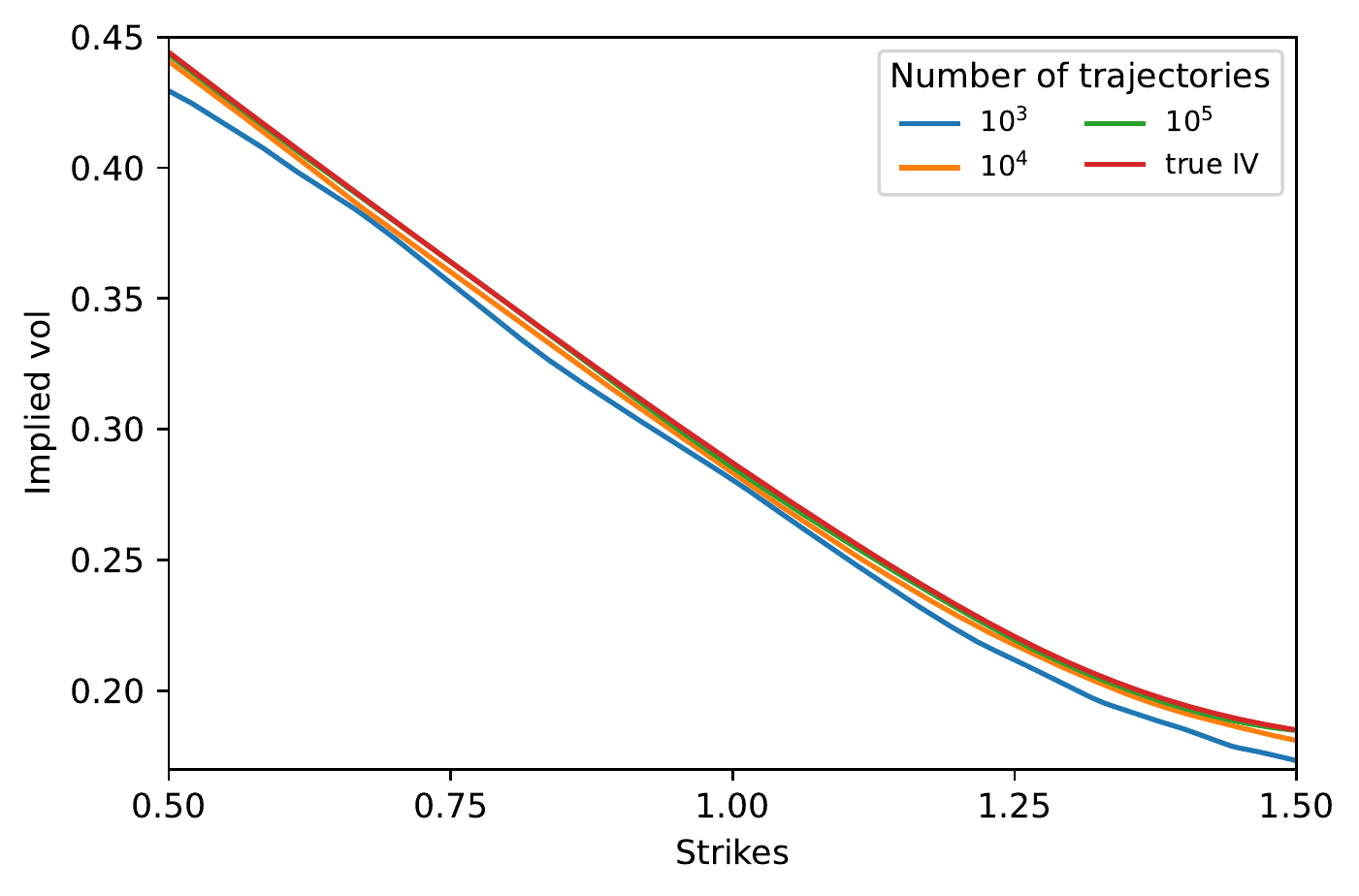}}
\vskip\baselineskip	\subfloat[]{\includegraphics[width=0.49\textwidth]{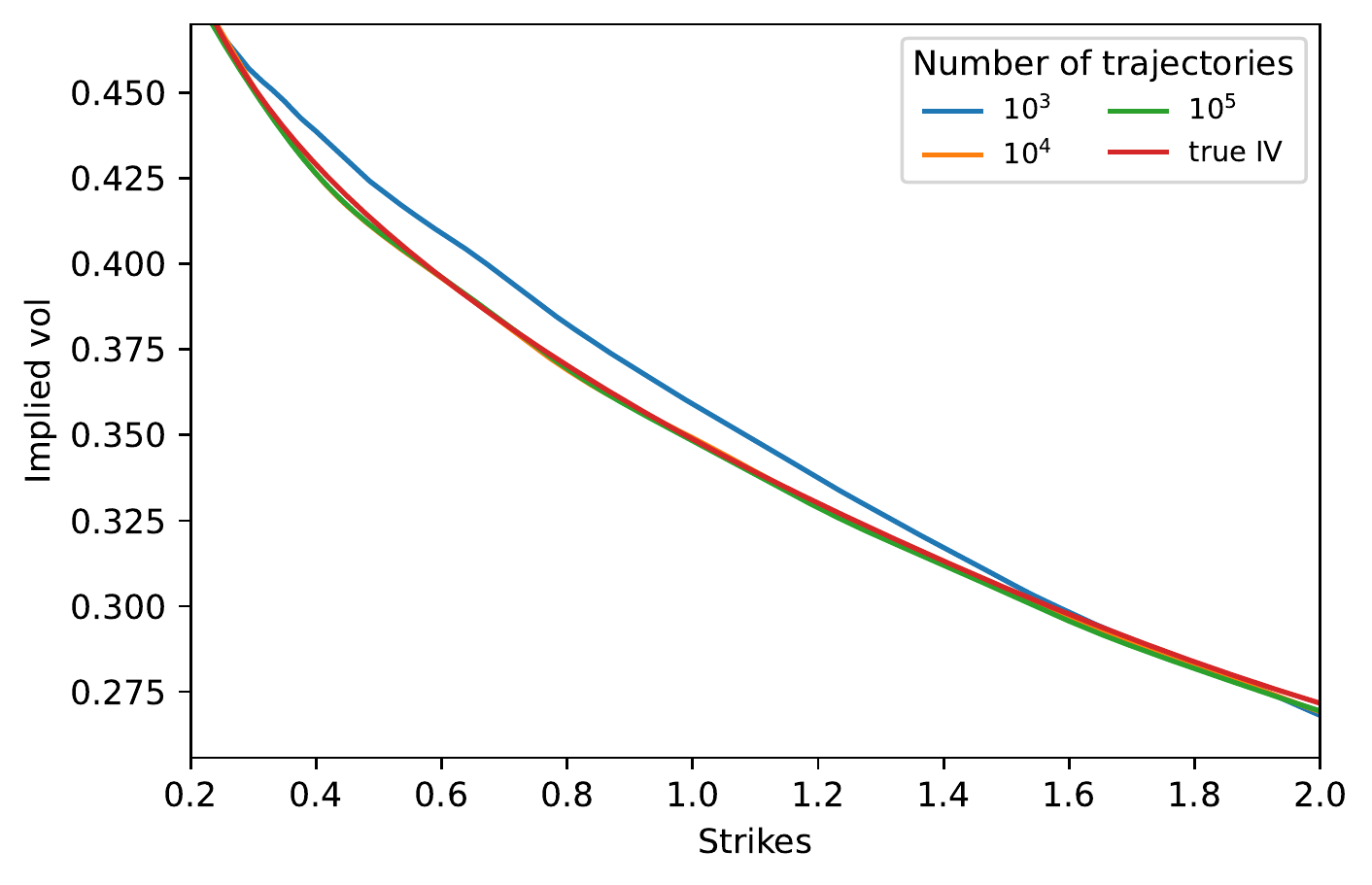}}
\hfill
\subfloat[]{\includegraphics[width=0.49\textwidth]{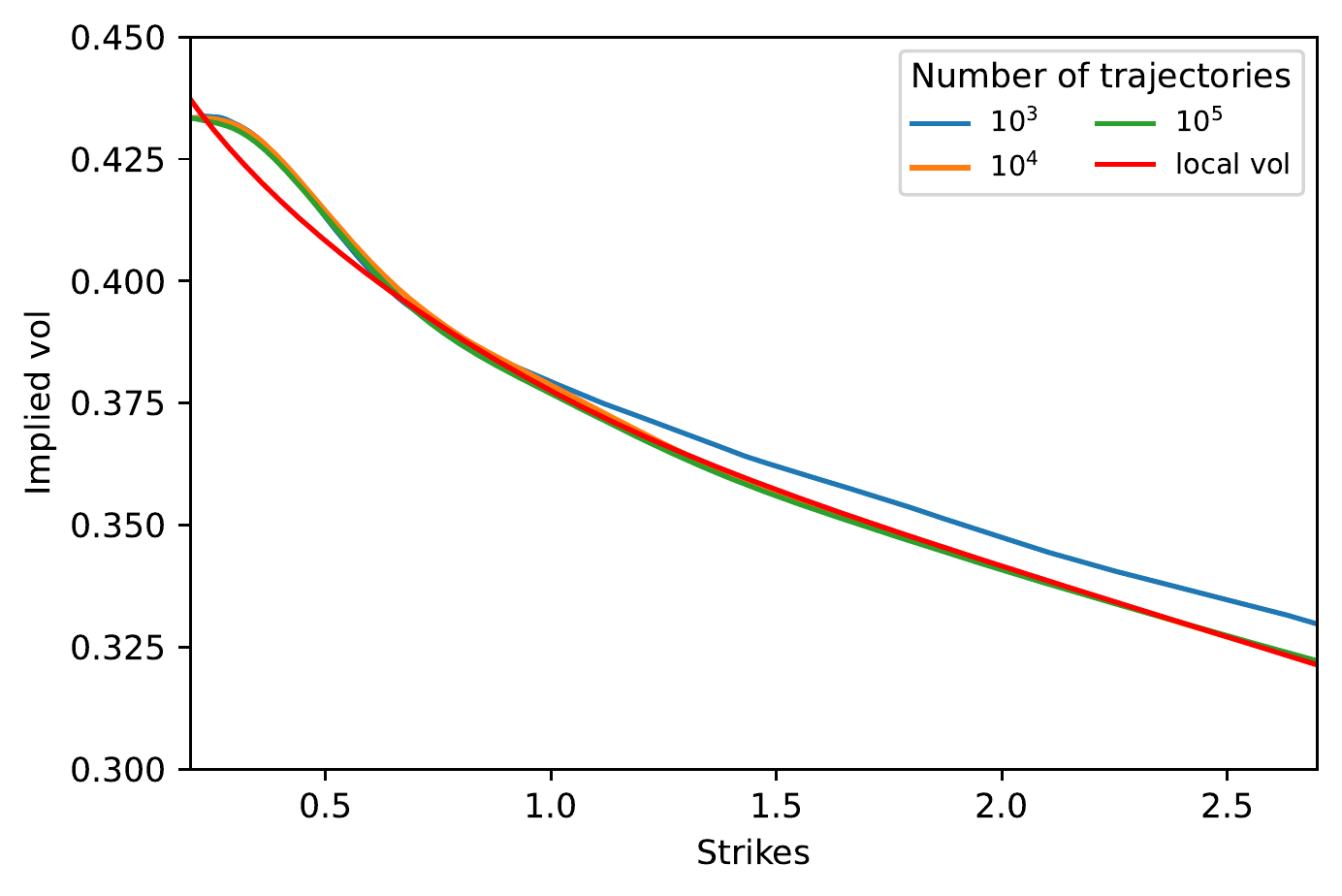}}
\caption{Fit of the smile for different number of particles.  (a): Black-Scholes setting, $T=1$ year. (b): Heston setting, $T=1$ year. (c): Heston setting, $T=4$ years. (d): Heston setting, $T=10$ years.
}\label{f:2}
\end{figure}

We plot the prices of the options on the logarithms of quadratic variation in \cref{f:QV}. It is immediate to see that in the Black-Scholes model (equation \eqref{processS} with $\sigma_{Dup}=\sigma$), we have $\langle \log S\rangle_T=\sigma^2 T$, and thus $\E[\langle \log S\rangle_T -K]_{+}=(\sigma^2 T-K)_{+}$. As shown in \cref{f:QV}(a), the prices of the options on the quadratic variation of $X$ are vastly different. This implies that despite the marginal distributions of $X$ and $S$ being identical, their dynamics are markedly dissimilar. We also see that these curves converges as the number of particles increases to infinity. This shows, that the dynamics of $(\wt X^n)$ is stable with respect to $n$. Options on the logarithms of quadratic variation for the Heston setting are presented in \cref{f:QV}(b). We see that, in this case, the dynamics of $X$ and $S$ are different, as expected and the dynamics of  $(\wt X^n)$ is also stable.

\begin{figure}
\centering
\subfloat[]{\includegraphics[width=0.49\textwidth]{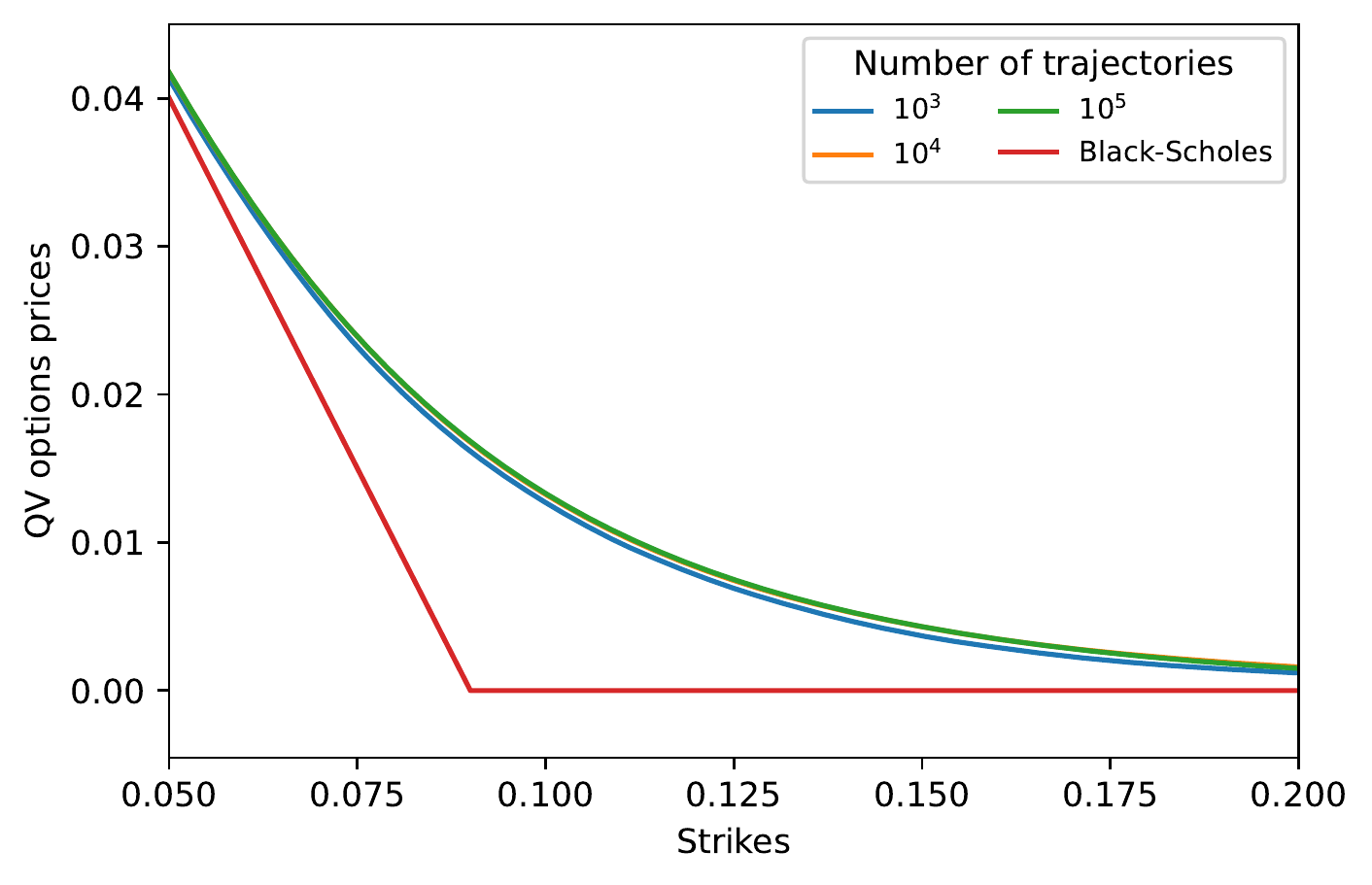}}
\hfill
\subfloat[]{\includegraphics[width=0.49\textwidth]{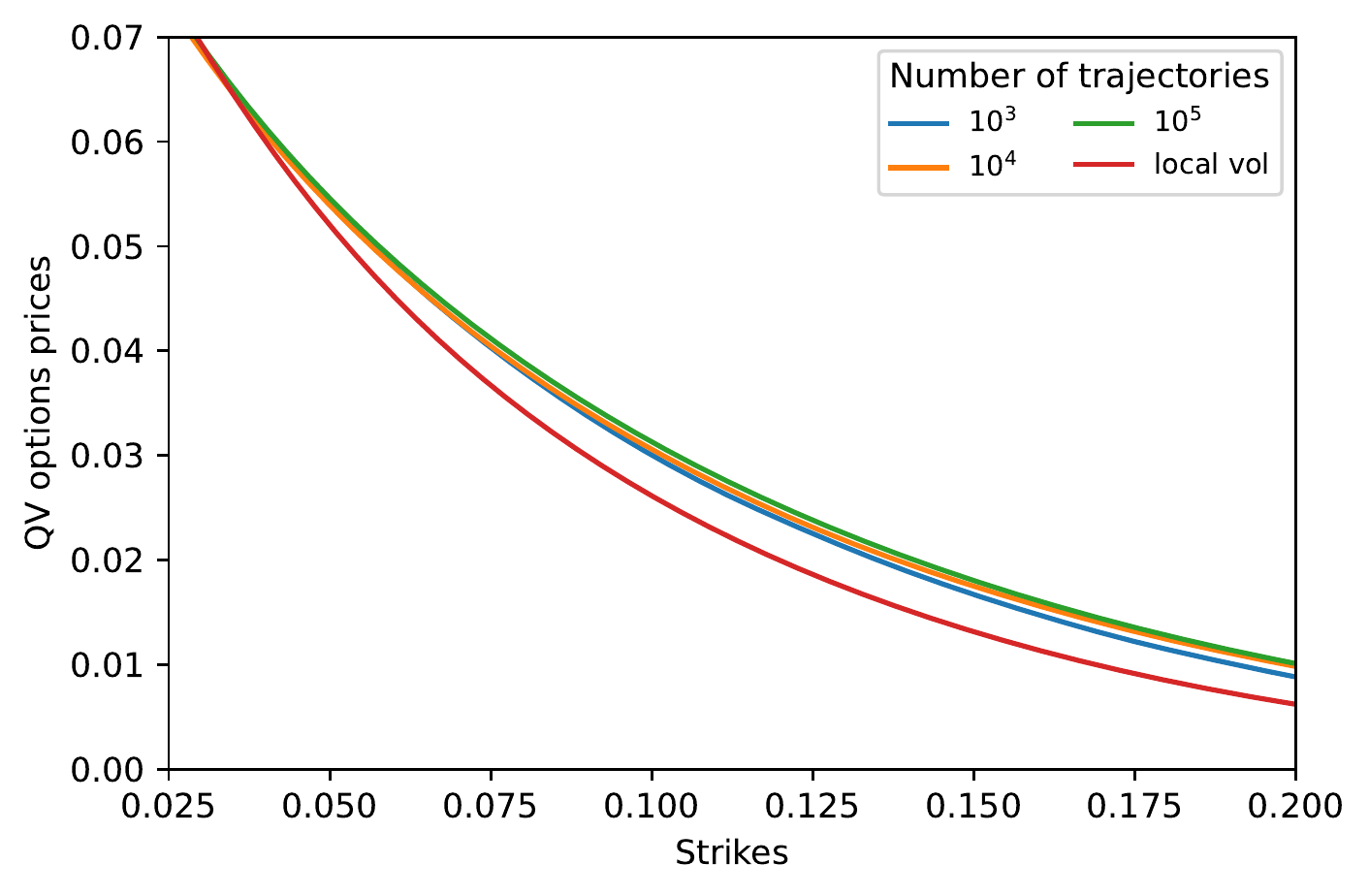}}
\caption{Prices of options on log of quadratic variation for different number of particles.  (a): Black-Scholes setting, $T=1$ year. (b): Heston setting, $T=1$ year.}\label{f:QV}
\end{figure}

It is interesting to compare our approach with the algorithm of \cite{G-HL} and \cite{guyon2011smile}. We consider a numerical setup similar to \cite[p.~10]{guyon2011smile}, taking $N=10^6$ particles to calculate implied volatilities. However, we calibrate our model and calculate the approximation of conditional expectation using only $N_1=1000$ of these particles. We compare our results in the Black-Scholes (a) and Heston (b) settings against implied volatilities calculated via the Euler method for the local volatility model $S$. \cref{f:new} shows great  agreement between the results of the two methods. 

\begin{remark}
\label{rem:runtime}
The computational time needed for running our algorithm is comparable with the algorithm of \cite{guyon2011smile}, but highly dependent on implementation details in both cases.
\end{remark}

\begin{figure}
\centering
\subfloat[]{\includegraphics[width=0.49\textwidth]{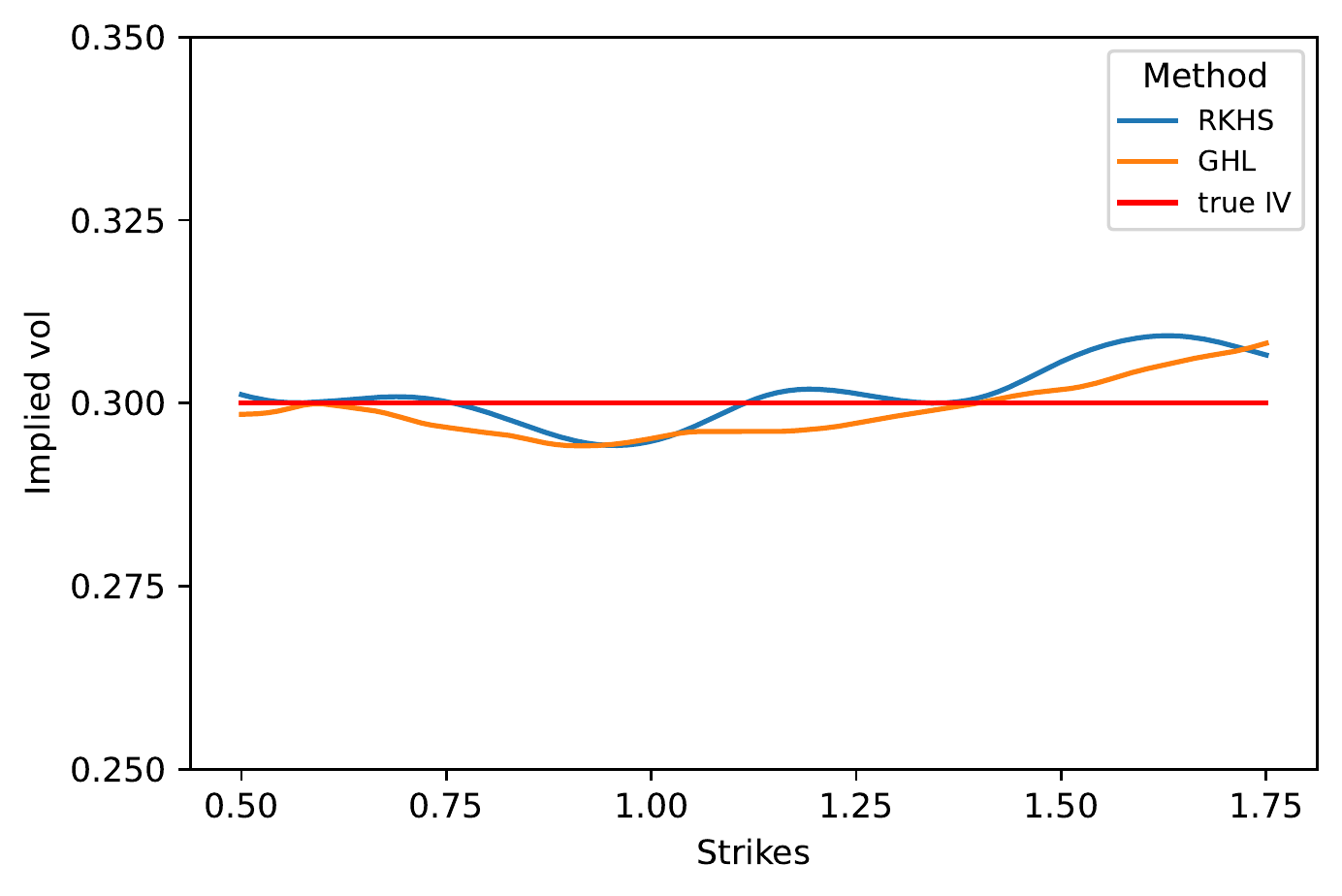}}
\hfill
\subfloat[]{\includegraphics[width=0.49\textwidth]{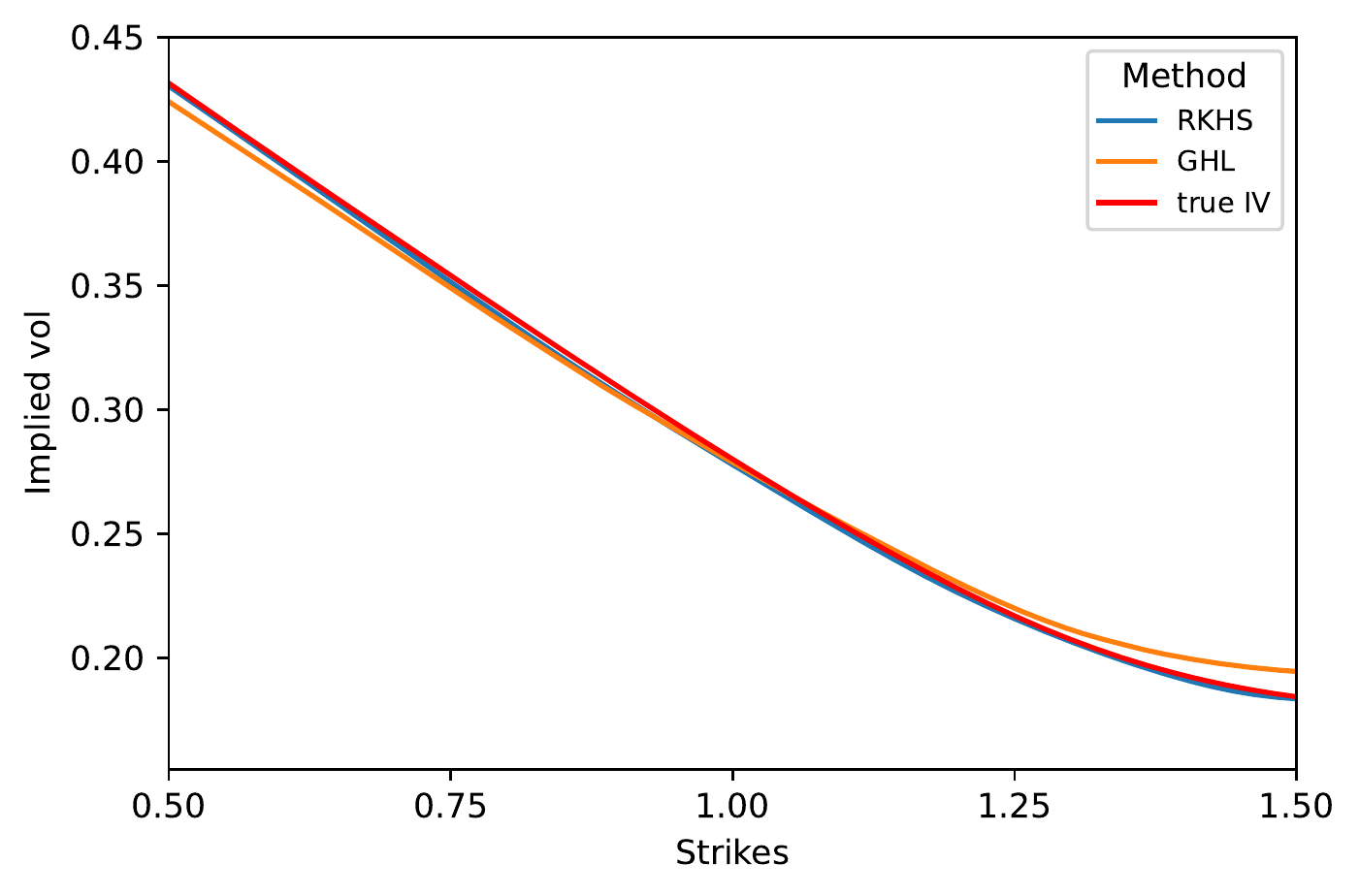}}
\caption{Comparison with \cite{G-HL}.  (a): Black-Scholes setting, $T=1$ year. (b): Heston setting, $T=1$ year.}\label{f:new}
\end{figure}

\cref{f:newd} shows that not only do the marginal distributions of $X$ calculated with our method and \cite{G-HL} agree with each other, but so do the cumulative distributions. We also observe that in both settings, the dynamics of $X$ are different from the dynamics of $S$.

\begin{figure}
\centering
\subfloat[]{\includegraphics[width=0.49\textwidth]{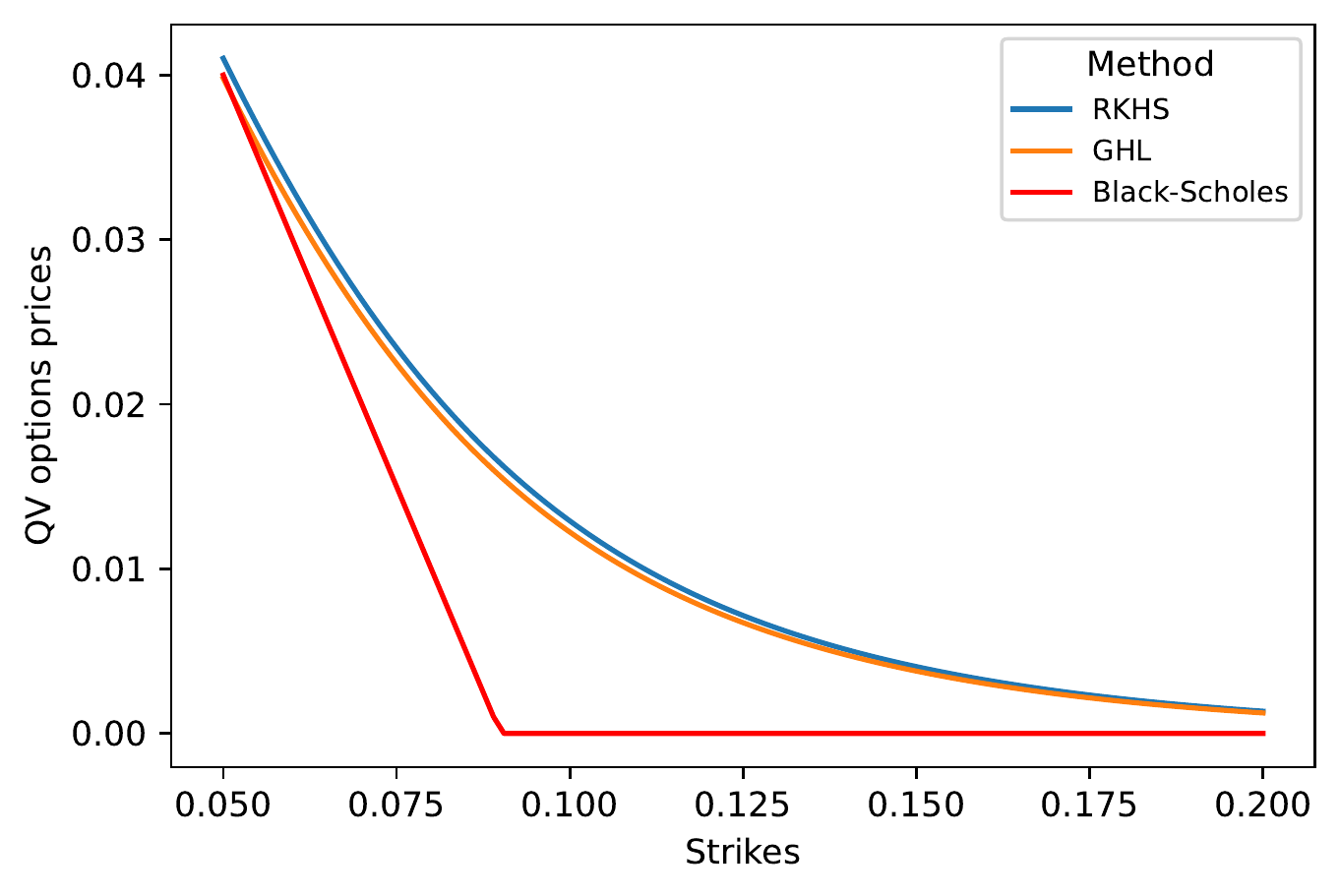}}
\hfill
\subfloat[]{\includegraphics[width=0.49\textwidth]{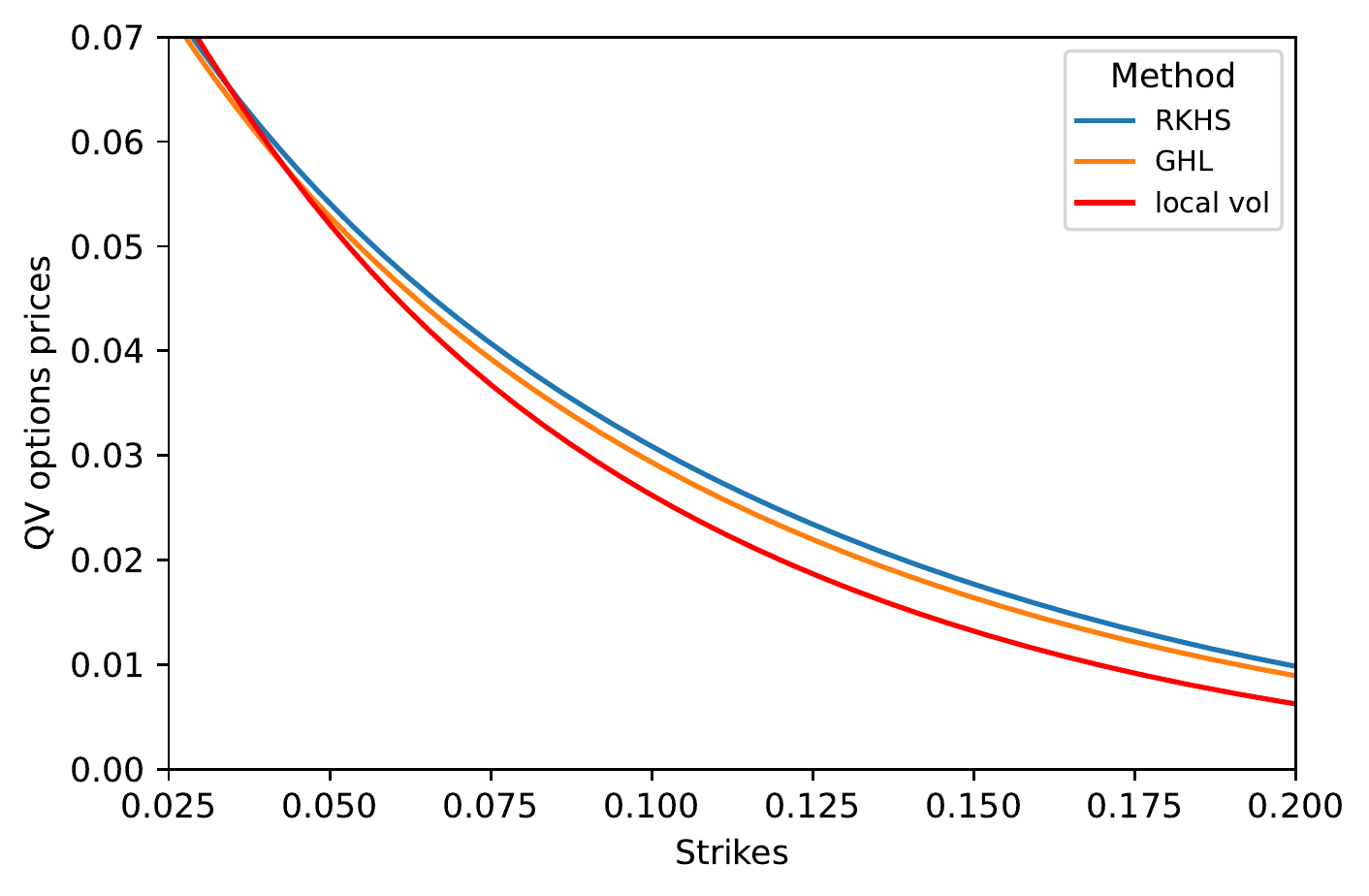}}
\caption{Comparison with \cite{G-HL}. Options on quadratic variation  (a): Black-Scholes setting, $T=1$ year. (b): Heston setting, $T=1$ year.}\label{f:newd}
\end{figure}

Now, let us discuss the stability of our model as the regularization parameter $\lambda\to0$. We studied the absolute error in the implied volatility of the 1-year ATM call option for various $\lambda\in[10^{-9},1]$ in the Black-Scholes and Heston settings described above. We used $N=10^6$ trajectories and $L=100$ $Z_j$s at each step according to \eqref{choice}, and performed $100$ repetitions at each considered value of $\lambda$. The results are presented in \Cref{f:0}.  The vertical lines in \cref{f:0}--\cref{f:3} denote the standard deviation in the absolute errors of the implied volatilities. We observe that in both settings, initially, the error drops as $\lambda$ decreases, then it stabilizes around $\lambda\approx 10^{-9}$. Therefore, for all of our calculations, we took $\lambda=10^{-9}$. It is evident that the error does not blow up as $\lambda$ becomes very small.

\begin{figure}
\centering
\subfloat[]{\includegraphics[width=0.49\textwidth]{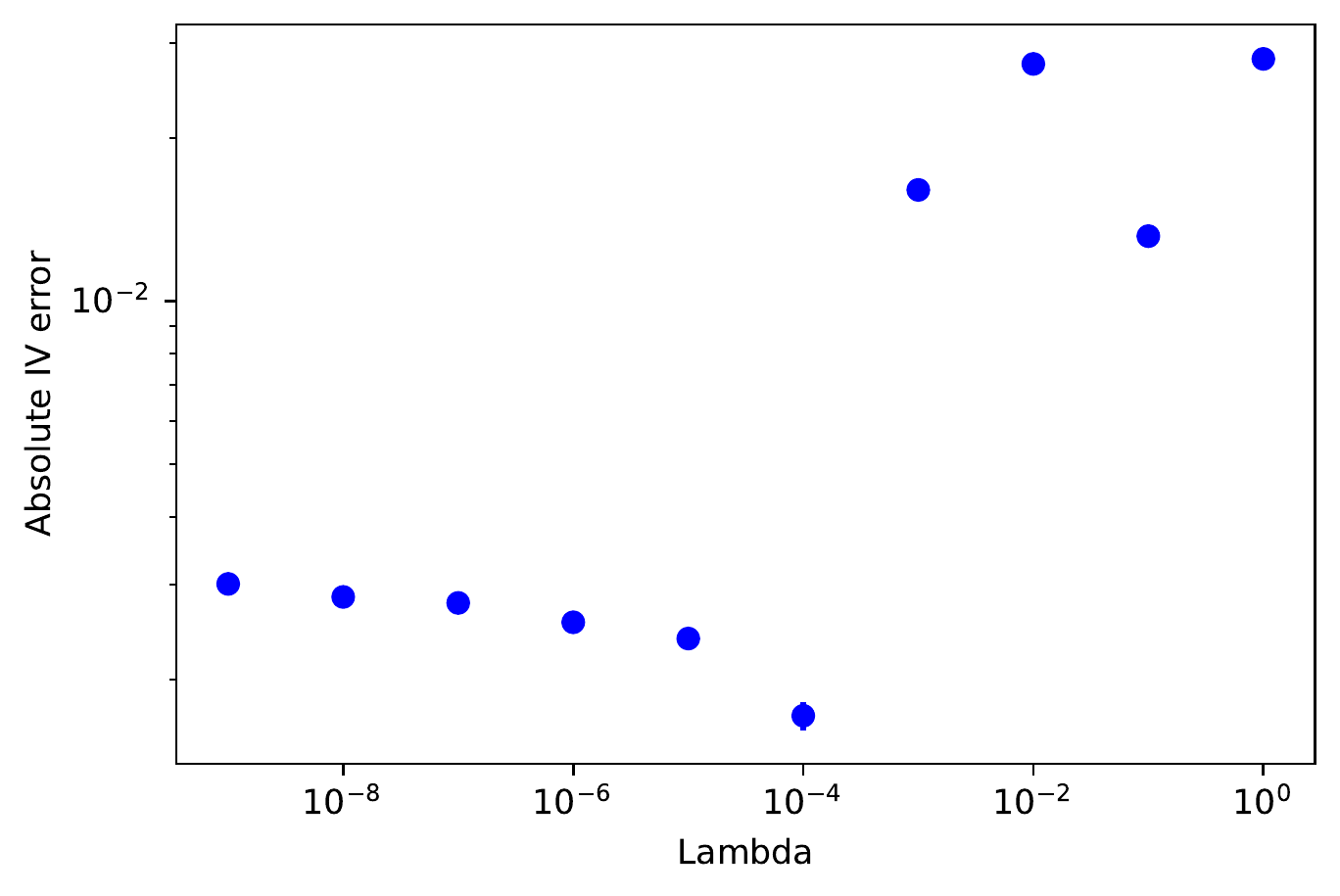}}
\hfill
\subfloat[]{\includegraphics[width=0.49\textwidth]{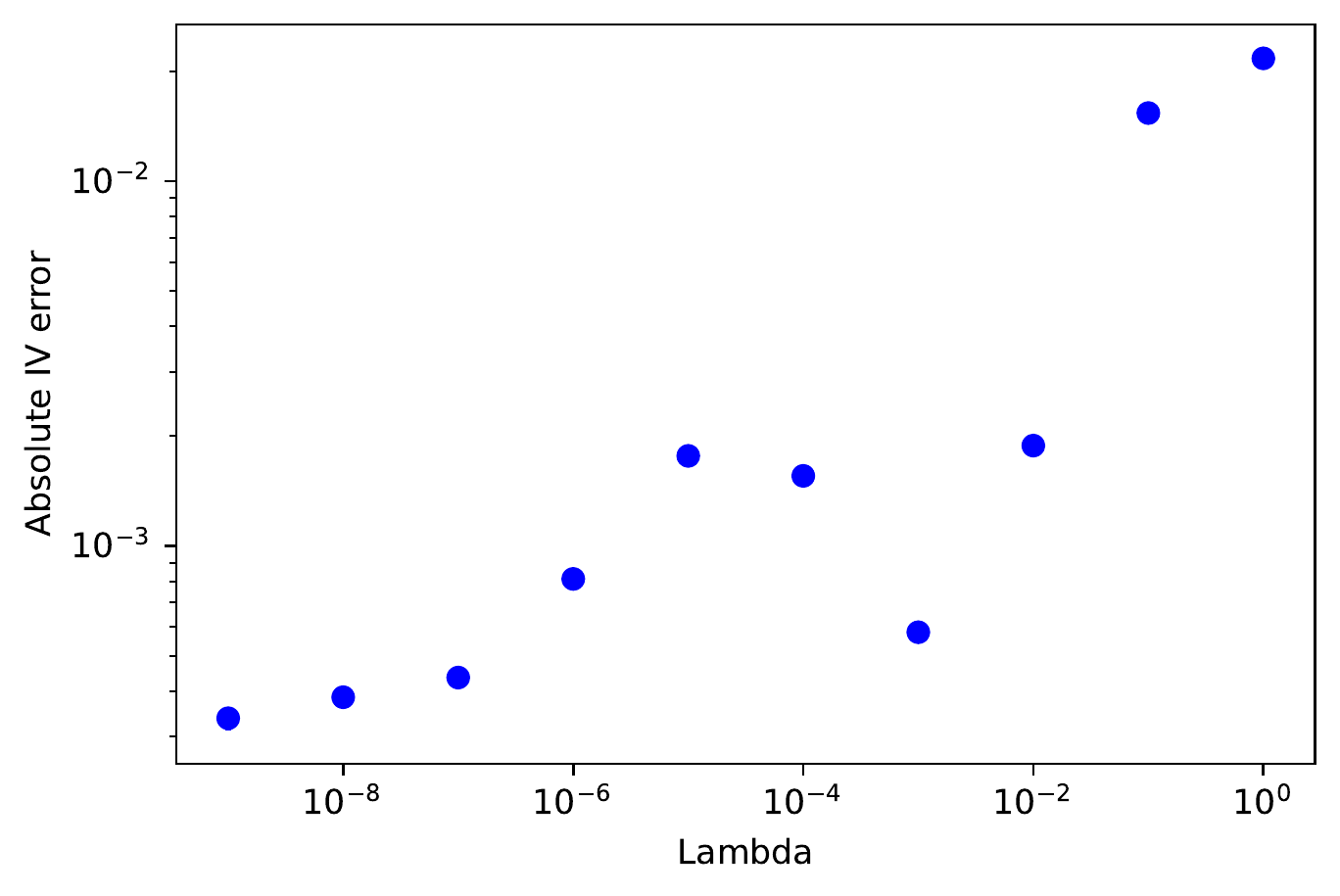}}
\caption{Mean absolute implied volatility error for different values of $\lambda$.  (a): Black-Scholes  setting. (b): Heston setting.
}\label{f:0}
\end{figure}

Let's examine how the error in call option prices in \eqref{priceidentity2} (and, therefore, the distance between the laws of the true and approximated solutions) depends on the number of trajectories $N$. Recall that it follows from \cref{T:MR2} that this error should decrease as $N^{-1/4}$ (note the square in the left-hand side of \eqref{eq:prop-chaos}). \Cref{f:1} shows how the absolute error in the implied volatility of a 1-year ATM call option decreases as the number of trajectories increases in (a) the Black-Scholes setting and (b) the Heston setting. We took $\lambda=10^{-9}$, $L=100$, $N\in[250, 2^{8}\cdot 250]$, and performed $100$ repetitions at each value of $N$. We see the error decreases as $O(N^{-1/2})$ in both settings, which is even better than predicted by theory.

\begin{figure}
\centering
\subfloat[]{\includegraphics[width=0.49\textwidth]{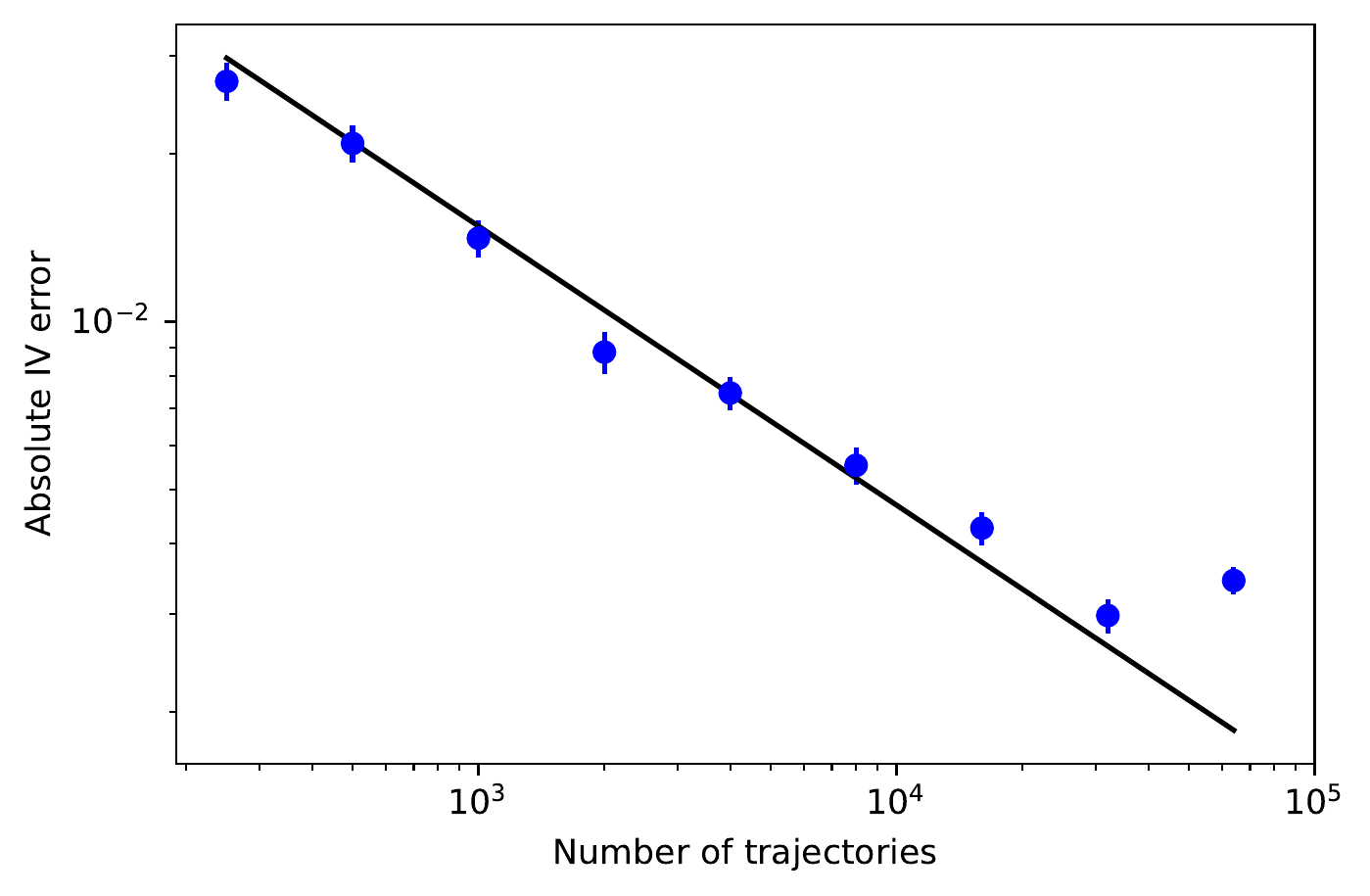}}
\hfill
\subfloat[]{\includegraphics[width=0.49\textwidth]{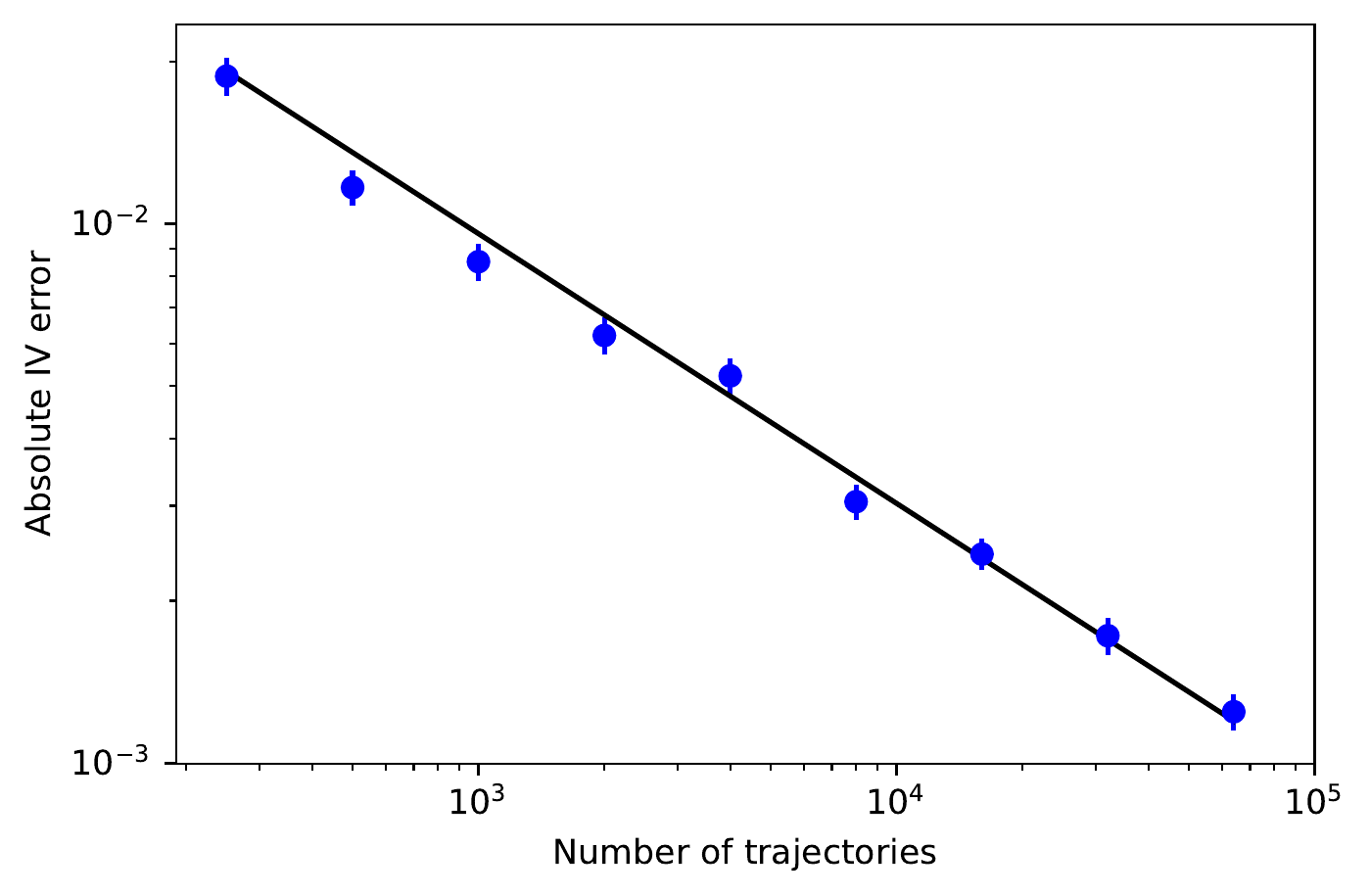}}
\caption{Mean absolute implied volatility error vs number of trajectories. The black line is the approximation: error$=C N^{-1/2}$  (a): Black-Scholes setting; $C=0.469$. (b): Heston setting; $C=0.303$.
}\label{f:1}
\end{figure}

We collect average errors in implied volatilities of 1 year European call options for different strikes in \cref{table:nonlin}. We considered the Heston setting and, as above, we used $\lambda=10^{-9}$, $L=100$, $N=10^5$.

\begin{table}[ht]
\centering
\begin{tabular}{|l| llllll |}
	\hline
	Strike $K$& 0.6& 0.8& 1& 1.2& 1.4& 1.6\\
	\hline
	$\P(S_T<K)$ & 0.0990& 0.2258& 0.4443& 0.7475&  0.9558& 0.9961\\
	True IV& 0.3999&0.3383&0.2795&0.2273&0.1927&0.1803\\
	IV error& 0.0011& 0.0018& 0.0006& 0.0032& 0.0043& 0.0011 \\
	\hline
\end{tabular}
\caption{Average error in implied volatility of 1 year options with given strike. Heston setting.}
\label{table:nonlin}
\end{table}

We also investigate the dependence of the error in the implied volatility on the number of basis functions $L$ in the representation \eqref{phidef}. Recall that since the number of operations depends on $L$ quadratically (it equals $O(MNL^2)$), it is extremely expensive to set $L$ to be large. In \Cref{f:3}, we plotted the dependence of the absolute error in the implied volatility of 1-year ATM call option on $L$. We used $N=10^6$ trajectories, $\lambda=10^{-9}$, $L\in[1\ldots100]$, and did $100$ repetitions at each value of the number of basis functions. We see that as the number of basis functions increases, the error first drops significantly but then stabilizes at $L\approx 80$.

\begin{figure}[H]
\centering
\subfloat[]{\includegraphics[width=0.49\textwidth]{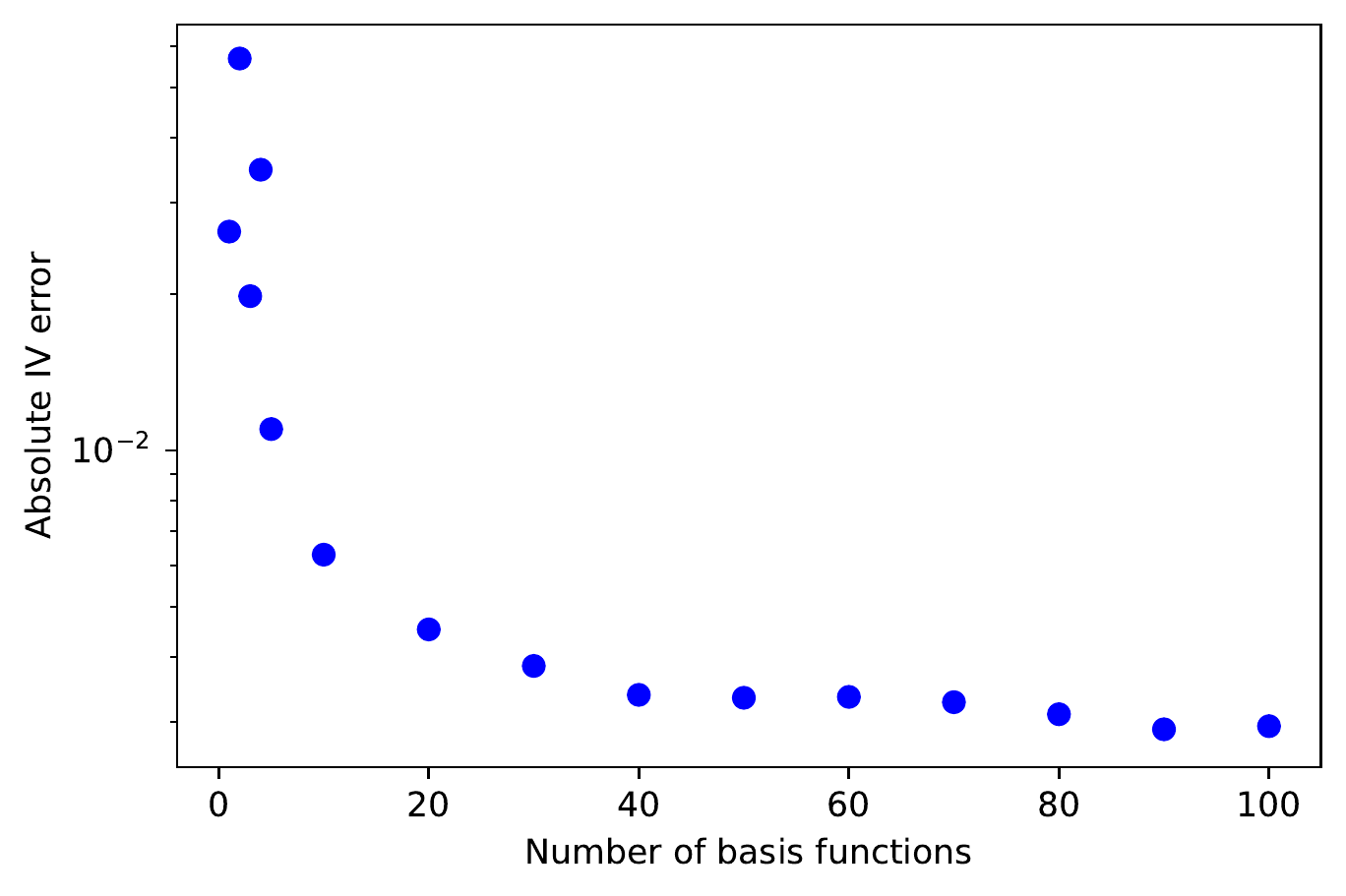}}
\hfill
\subfloat[]{\includegraphics[width=0.49\textwidth]{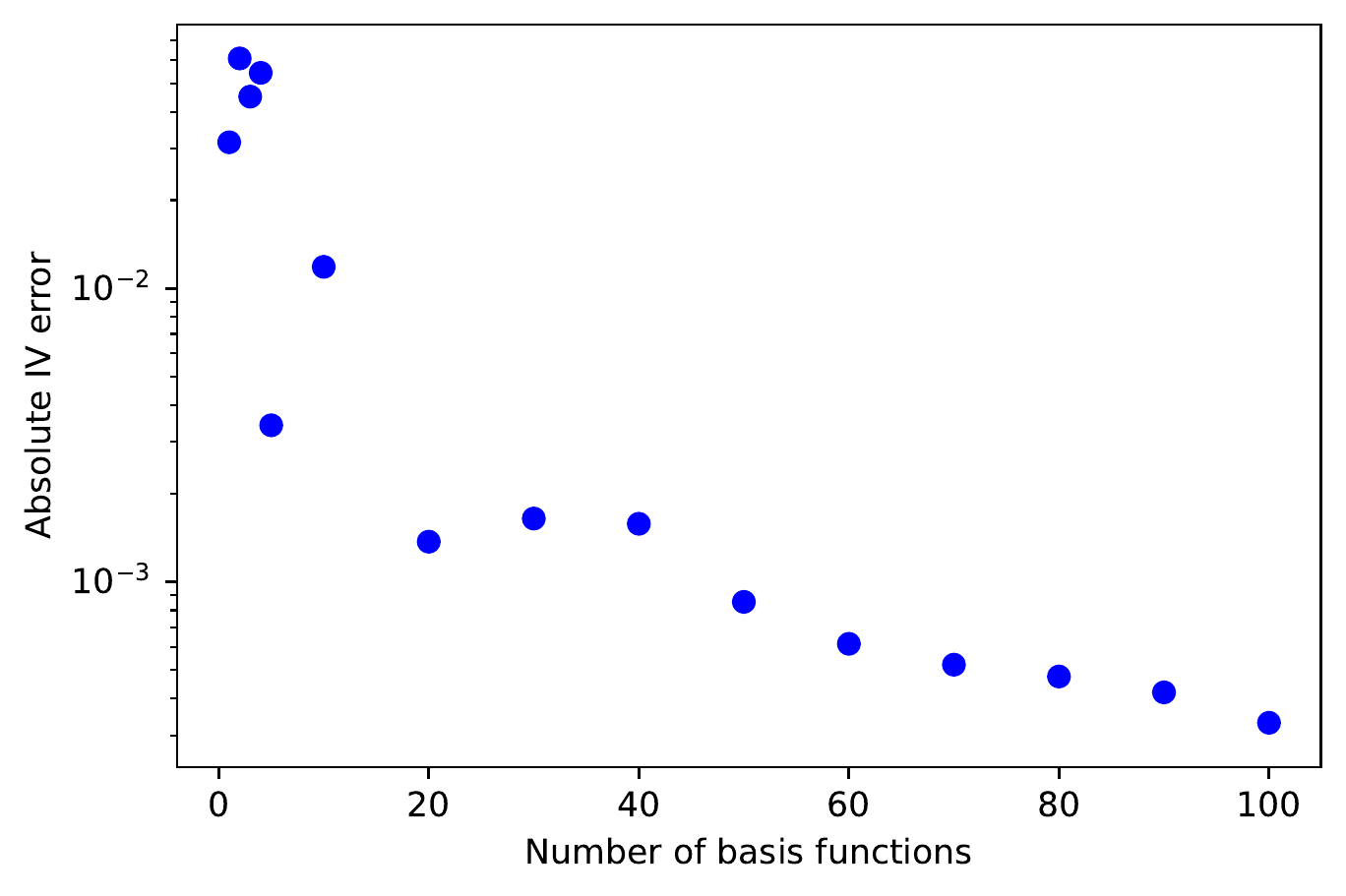}}
\caption{Mean absolute implied volatility error  vs number of basis functions.  (a): Black-Scholes  setting. (b): Heston setting.
}\label{f:3}
\end{figure}

\subsection{On the choice of $(\varepsilon,\varepsilon_{CIR})$}

We recall that there are two different truncations involved in the model. First, we cap the CIR process from below at the level of  $\eps_{CIR}=10^{-3}$. Second, in the Euler scheme \eqref{eq:Euler-mvsdea}--\eqref{eq:Euler-mvsdeb} we take as a diffusion
$$
F(t,x,y,z):=x \sigma_{\text{Dup}}(t,x)\frac{\sqrt y}{\sqrt {z\vee\eps} },
$$
with $\eps=10^{-3}$. We claim that both of this truncations are necessary.

\begin{figure}[h]
\centering
\subfloat[]{\includegraphics[width=0.49\textwidth]{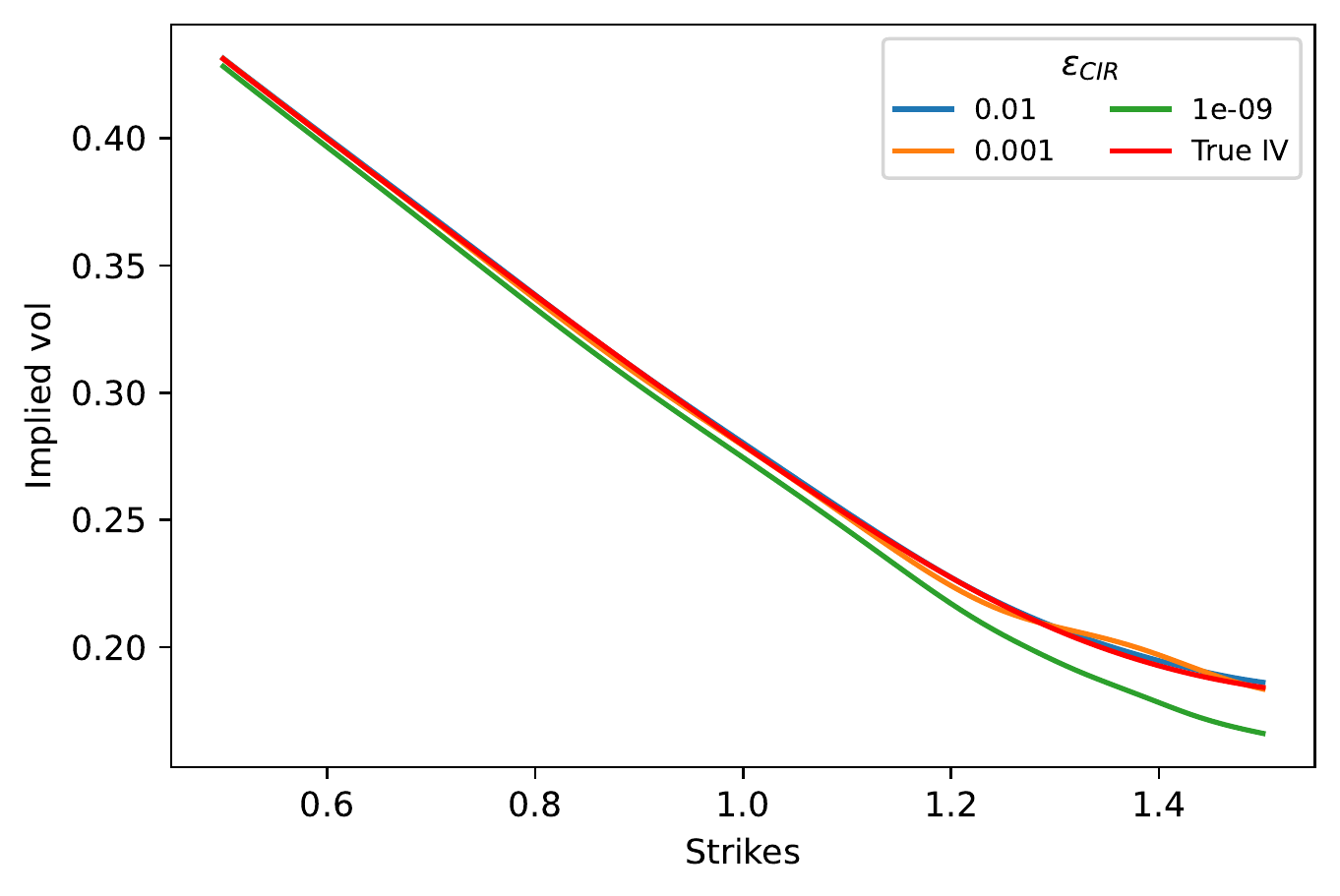}}
\hfill
\subfloat[]{\includegraphics[width=0.49\textwidth]{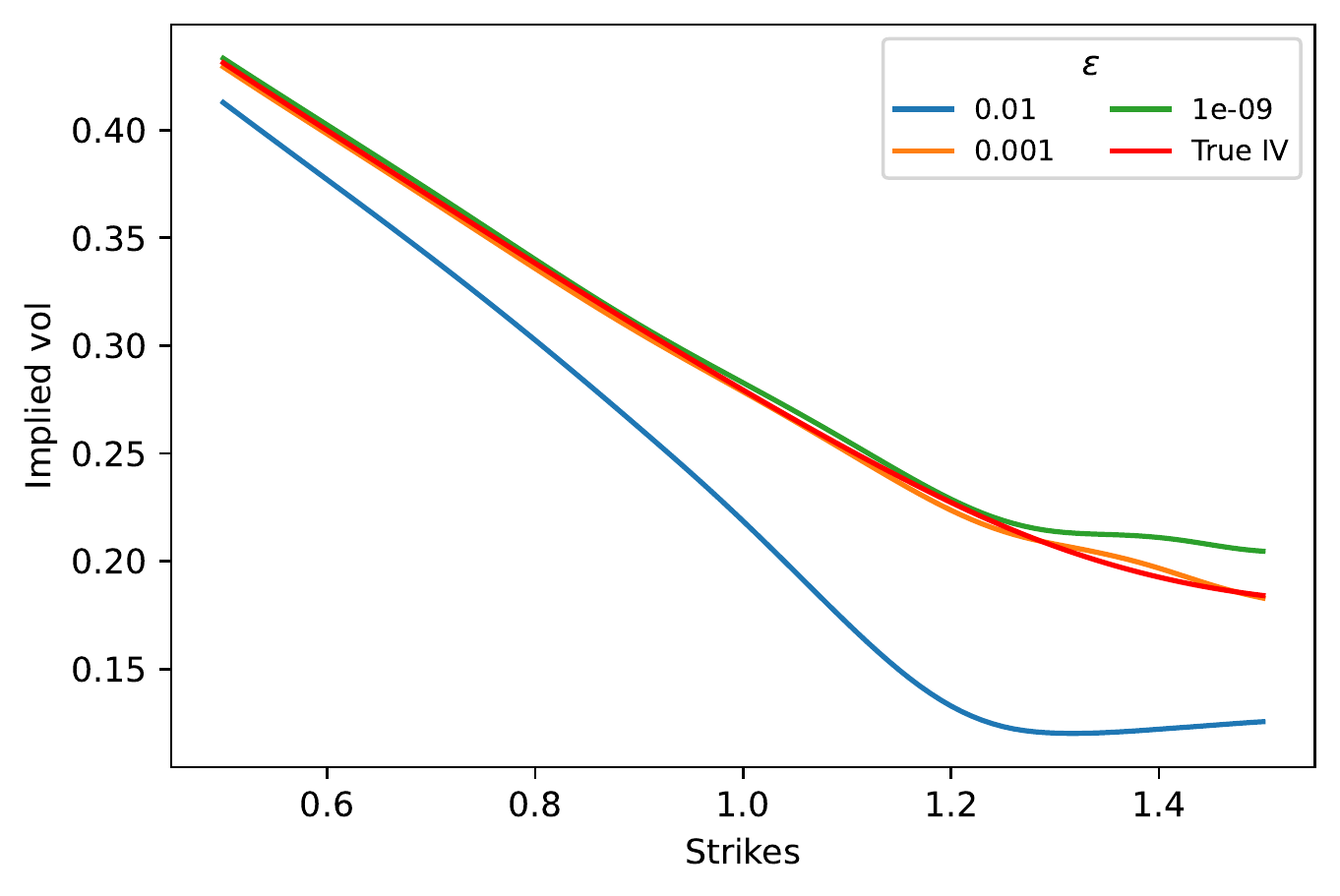}}
\hfill
\subfloat[]{\includegraphics[width=0.49\textwidth]{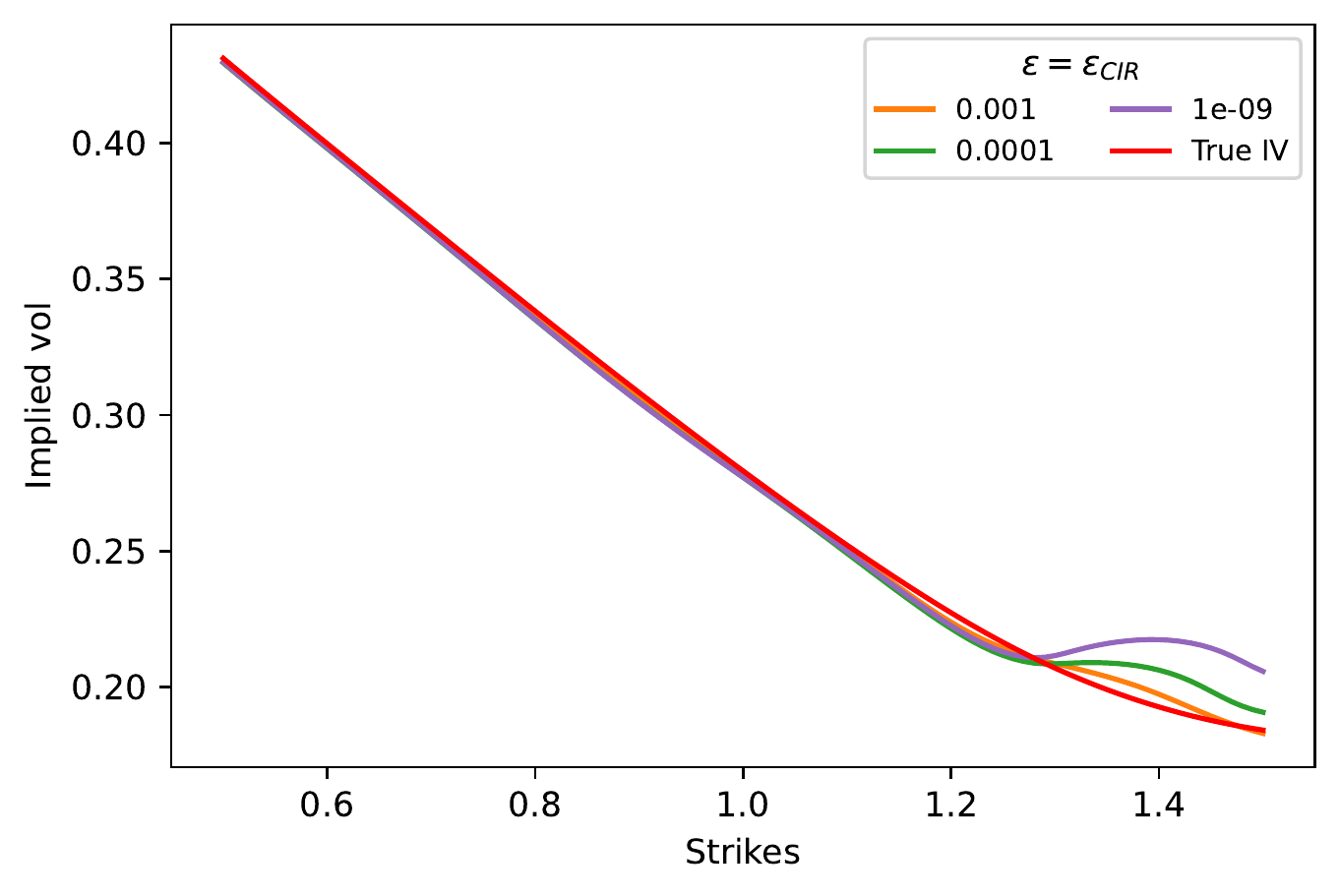}}		
\caption{Fit of the smile of 1-year call options for different truncation levels  (a): $\eps=10^{-3}$, $\eps_{CIR}$ varies; (b): $\eps_{CIR}=10^{-3}$, $\eps$ varies;
	(c): $\eps=\eps_{CIR}$ varies.
}\label{f:100}
\end{figure}
\cref{f:100} below shows fit of the smile  for 1-year European call options depending on  $\eps$ and $\eps_{CIR}$. We use the model of Section~5   with $M=500$ timesteps and $N=10^6$ trajectories. We see from these plots that if $\eps$ or $\eps_{CIR}$ are either too small or too large, the smile produced by the model may not closely match the true implied volatility curve. Therefore, a certain lower capping of the CIR process is indeed necessary.

\section{Conclusion and outlook}
\label{S:6}

In this paper, we study the problem of calibrating local stochastic volatility models via the particle approach pioneered in \cite{G-HL}. We suggest a novel RKHS based regularization method and prove that this regularization guarantees well-posedness of the underlying McKean-Vlasov  SDE and the propagation of chaos property. 
Our numerical results suggest that the proposed approach is rather efficient for the calibration of various local stochastic volatility models and can obtain similar efficiency as widely used local regression methods, see \cite{G-HL}.
There are still some questions left open here. First, it remains unclear whether the regularized McKean-Vlasov  SDE remains well-posed when the regularization parameter \(\lambda\) tends to zero. This limiting case needs a separate study. Another important issue is the choice of RKHS and the number of basis functions which ideally should be adapted to the problem at hand. This problem of adaptation is left for future research.

\section{Proofs}
\label{sec:proof}
In this section we present the proofs of the results from \cref{sec:mkv} and \cref{sec:rkhs}.
\begin{proof}[Proof of \cref{prop: apcon}]
Since  $\mathcal{H}$ is separable, let $I\subset\N$ and let
$e:=(e_{i})_{i\in I}$ be a total orthonormal system in $\mathcal{H}$ (note that $I$ is finite if $\mathcal{H}$ is finite dimensional). Define
the vector $\gamma^{\nu}\in\ell_{2}(I)$ by
\begin{align}
	\gamma^{\nu}_i:=\langle e_{i},c^{\nu}_A\rangle _{\H}  &  =\int_{\X\times\X}\big\langle
	e_{i},k(\cdot,x)\big\rangle _{\mathcal{H}}A(y)\nu(dx,dy)\notag \\
	&  =\int_{\X\times\X} e_{i}(x)A(y)\nu(dx,dy),\quad i\in I.\label{gdef}
\end{align}
Since the operator $\mathcal{C}^{\nu}$ is bounded it may
be described by the (possibly infinite) symmetric matrix%
\begin{equation}
	B^{\nu}:=\big(  \langle e_{i},\mathcal{C}^{\nu}e_{j}\rangle_{\H}
	\big)  _{(i,j)\in I\times I}=\Big(  \int_\X e_{i}(x)e_{j}(x)\,\nu
	(dx,\mathcal{X})\Big)  _{(i,j)\in I\times I}, \label{Bdef}%
\end{equation}
which acts as a bounded positive semi-definite operator on $\ell_{2}(I)$. Denote
\begin{equation}
	\beta^{\nu}=(B^{\nu}+\lambda I)^{-1}\gamma^{\nu}. \label{betanu}
\end{equation}
For $f\in \H$ write $f=\sum_{i\in I} \beta_i e_i$. Then, recalling \eqref{gdef} and \eqref{Bdef}, we derive
\begin{align*}
	&\argmin_{f\in\H}\bigl\{  \int_{\X\times\X}|A(y)-f(x)|^2\,\nu(dx,dy)+\lambda
	\| f\|_{\H}^2\bigr\}\\
	&\quad=\argmin_{\beta\in\ell_2(I)}\bigl\{  \int_{\X\times\X}|A(y)-\sum_{i\in I} \beta_i e_i|^2\,\nu(dx,dy)+\lambda
	\|\beta\|_{\ell_2(I)}^2\bigr\}\\
	&\quad=\argmin_{\beta\in\ell_2(I)}\bigl\{ -2 \langle \beta,\gamma^\nu\rangle_{\ell_2(I)}+
	\langle \beta,(B^{\nu}+\lambda I) \beta\rangle_{\ell_2(I)} \bigr\} \\
	&\quad=\argmin_{\beta\in\ell_2(I)}\bigl\{  \langle  \beta-\beta^\nu,(B^{\nu}+\lambda I)  (\beta-\beta^\nu)\rangle_{\ell_2(I)}  \bigr\}\\
	&\quad= \beta^\nu,
\end{align*}
where  we inserted definition \eqref{betanu} and used the fact that $B^{\nu}+\lambda I$ is strictly positive definite for $\lambda>0$.
To complete the proof it remains to note that
\begin{equation*}
	\sum_{i=1}^\infty \beta^\nu_i e_i=(\mathcal{C}^{\nu}+\lambda I_{\H})^{-1}c^\nu_A,
\end{equation*}
which shows \eqref{identm}.
\end{proof}

\begin{proof}[Proof of \cref{lip}]
Let us write
\begin{align}\label{step0}
	|m_A^{\lambda}(x;\mu)-m_A^{\lambda}(y;\nu)| &
	\le | m_A^{\lambda}(x;\mu)-m_A^{\lambda}(x;\nu)|
	+| m_A^{\lambda}(x;\nu)-m_A^{\lambda}(y;\nu)|\nn
	\\
	&  =I_{1}+I_{2}.
\end{align}
Working with respect to the orthonormal basis introduced in the proof of \cref{prop: apcon}, see
\eqref{betanu}, we derive for the first term in \eqref{step0}
\begin{align}\label{Ione}
	I_{1} &  =|\langle k(x,\cdot),m^{\lambda}_A(\cdot;\mu)-m^{\lambda}_A(\cdot;\nu)\rangle_{\H}|\nn\\
	&\le \|k(x,\cdot)\|_{\H}\|m^{\lambda}_A(\cdot;\mu)-m^{\lambda}_A(\cdot;\nu)\|_{\H}\nn\\
	&  \leq\sqrt{k(x,x)}\| \beta^{\mu}-\beta^{\nu}\|_{\ell_{2}(I)}\nn\\
	&  \le D_k \| \beta^{\mu}-\beta^{\nu}\|_{\ell_{2}(I)}
\end{align}
where we used \eqref{knorm} and Assumption~\ref{AKnew}.

Denote $Q^\nu:=B^{\nu}+\lambda I$ and $Q^{\mu}:= B^{\mu}+\lambda I$. Recalling that they  are bounded $\ell_{2}(I)\to\ell_{2}(I)$ operators with bounded
inverses, it easy to see that
\begin{equation*}
	\| (Q^{\mu})^{-1}-(Q^{\nu})^{-1}\|
	_{\ell_{2}(I)}\leq\| (Q^{\mu})^{-1}\|_{\ell_{2}(I)}\| (Q^{\nu})^{-1}\|_{\ell_{2}(I)}\|Q^\mu-Q^{\nu}\|_{\ell_{2}(I)}.
\end{equation*}
Therefore
\begin{align}
	\| \beta^{\mu}-\beta^{\nu}\|_{\ell_{2}(I)}  &
	=\|(Q^{\mu})^{-1}\gamma^{\mu}-(Q^{\nu})^{-1}\gamma^{\nu}\|_{\ell_{2}(I)}\nonumber\\
	&  \le\big\|\big(  (Q^{\mu})^{-1}-(Q^{\nu})^{-1}\big)  \gamma^{\mu}\big\|
	_{\ell_{2}(I)}+\big\| (Q^{\nu})^{-1}(  \gamma^{\mu}-\gamma^{\nu})  \big\|_{\ell_{2}(I)}\nonumber\\
	&  \leq\big\| (Q^{\mu})^{-1}\big\| _{\ell_{2}(I)}\big\| (Q^{\nu})^{-1}\big\|_{\ell_{2}(I)}\big\| Q^\mu-Q^{\nu}\big\| _{\ell_{2}(I)}\big\Vert \gamma^{\mu}\big\| _{\ell_{2}(I)}\nn\\
	&\quad+\big\| (Q^{\nu})^{-1}\big\Vert _{\ell_{2}(I)}\big\Vert \gamma^{\mu}-\gamma^{\nu}\big\Vert _{\ell_{2}(I)}\nonumber\\
	&  \leq\frac{1}{\lambda^{2}}\big\| B^{\mu}-B^{\nu}\big\|_{\ell_{2}(I)}\big\Vert \gamma^{\mu}\big\Vert _{\ell_{2}(I)}+\frac{1}{\lambda
	}\big\| \gamma^{\mu}-\gamma^{\nu}\big\|_{\ell_{2}(I)}.
	\label{no}%
\end{align}
Now observe that for any $i,j\in I$
\begin{align*}
	(  B_{ij}^{\mu}-B_{ij}^{\nu})  ^{2}  &  =\Big(  \int_{\X}
	e_{i}(x)e_{j}(x)\big(  \mu(dx,\mathcal{X})-\nu(dx,\mathcal{X}%
	)\big)  \Big)  ^{2}\\
	&  =\int_{\X}\int_{\X} e_{i}(x)e_{j}(x)e_{i}(y)e_{j}(y)\\
	&  \quad\times\big(  \mu(dx,\mathcal{X})-\nu(dx,\mathcal{X})\big)
	\big(  \mu(dy,\mathcal{X})-\nu(dy,\mathcal{X})\big)  .
\end{align*}
Hence, by using the identity
\begin{equation}\label{sumbasis}
	\sum_{i\in I}e_{i}(x)e_{i}(y)=\sum_{i\in I}\big\langle k(x,\cdot),e_{i}
	\big\rangle_{\H} \big\langle k(y,\cdot),e_{i}\big\rangle_{\H} =\big\langle  k(x,\cdot)
	, k(y,\cdot)\big\rangle _{\mathcal{H}}=k(x,y),
\end{equation}
we get
\begin{align}\label{Bdiff}
	\big\|B^{\mu}-B^{\nu}\big\|_{\ell_{2}(I)}^{2}&\le\big\| B^{\mu}-B^{\nu}\big\|_{HS}^{2}\nn\\
	&=\int_{\X}\big(
	\mu(dx,\mathcal{X})-\nu(dx,\mathcal{X})\big)  \int_{\X} k^{2}
	(x,y)\big(  \mu(dy,\mathcal{X})-\nu(dy,\mathcal{X})\big).
\end{align}
By the  Kantorovich-Rubinstein duality formula (\cite{villani2021topics}),  for every  $h:\mathcal{X}\rightarrow$ $\mathbb{R}$ with \(h\in C^1(\mathcal{X})\)  one has
\begin{align*}
	\Big| \int_{\X} h(x)\big(  \mu(dx,\mathcal{X})-\nu (dx,\mathcal{X}%
	)\big)  \Big|  &  =\Big| \int_{\X\times\X} h(x)\big(  \mu(dx,dy)-\nu
	(dx,dy)\big)  \Big| \\
	&  \leq\sup_{x\in\mathcal{X}}\big\vert \partial_x h(x)\big\vert \W_{1}(\mu
	,\nu),
\end{align*}
where $\partial_x$ denotes gradient with respect to $x.$ So we continue \eqref{Bdiff} in the following way:
\begin{equation}\label{Bdiff2}
	\|B^{\mu}-B^{\nu}\|_{\ell_2(I)}^{2}\leq \W_{1}(\mu
	,\nu)\sup_{x\in\mathcal{X}}\Big| \int_{\X}\partial_{x}k^{2}%
	(x,y)\big(  \mu(dy,\mathcal{X})-\nu(dy,\mathcal{X})\big)
	\Big|,
\end{equation}
and for each particular $x\in\mathcal{X}$ we have similarly
\begin{align*}
	&\Bigl|\int_{\X}\partial_{x}k^{2}(x,y)\big(\mu(dy,\mathcal{X})-\nu
	(dy,\mathcal{X})\big)\Bigr|\\
	&\qquad \leq\sum_{i=1}^{d}\Bigl|\int_{\X}\partial_{x_{i}}k^{2}(x,y)\big(\mu
	(dy,\mathcal{X})-\nu(dy,\mathcal{X})\big)\Bigr|\\
	&\qquad \leq\sum_{i=1}^{d}\sup_{y\in\mathcal{X}}|\partial_{y}\partial
	_{x_{i}}k^{2}(x,y)|\W_{1}(\mu,\nu)\\
	&\qquad\leq d^{2}D_{k}^{2}\W_{1}(\mu,\nu),
\end{align*}
where the last inequality follows from by Assumption~\ref{AKnew}. Combining this with \eqref{Bdiff2}, we deduce
\begin{equation}
	\|B^{\mu}-B^{\nu}\|_{\ell_2(I)} \leq  D_{k}\W_{1}(\mu,\nu)d. \label{nb}%
\end{equation}
By a similar argument, using \eqref{sumbasis}, we derive
\begin{align}
	&\big\| \gamma^{\mu}-\gamma^{\nu}\big\|_{\ell_{2}(I)}^{2}\nn\\
	&\quad\le \sum_{i\in I} \int_{\X\times \X}\int_{\X\times \X} e_i(x)e_i(x')A(y)A(y')(\mu-\nu)(dx,dy)(\mu-\nu)(dx',dy')\nn\\
	&\quad\le  \int_{\X\times \X}\int_{\X\times \X} k(x,x')A(y)A(y')(\mu-\nu)(dx,dy)(\mu-\nu)(dx',dy')\nn\\
	&\quad
	\leq d^2\W_{1}^{2}(\mu, \nu)\| A\|_{\C^1}^{2}D_{k}^{2} \label{ng},
\end{align}
where again Assumption~\ref{AKnew} was used. Next note that
\begin{align}
	\| \gamma^{\mu}\|_{\ell_{2}(I)}^{2}   &=\int_{\X\times \X}\int_{\X\times \X} k(x,x')A(y)A(y')\mu(dx,dy)\mu(dx',dy')\nonumber\\
	&  \leq\int_{\X\times \X}\int_{\X\times \X}\big| A(y)\big| \sqrt{k(x,x)}\big|
	A(y')\big| \sqrt{k(x',x')}\mu
	(dx,dy)\mu(dx',dy')\nonumber\\
	&  =\Bigl(  \int_{\X\times \X}\big| A(y)\big| \sqrt{k(x,x)}\mu(dx,dy)\Bigr)^{2}\nonumber\\
	&  \leq\int_{\X\times \X}\big| A(y)\big| ^{2}\mu(dx,dy)\int_{\X\times \X} k(x,x)\mu(dx,dy)\nn\\
	&\leq
	D_{k}^{2}\| A\|_{\C^1}^{2} \label{bga}%
\end{align}
due to Assumption~\ref{AKnew}. Substituting now \eqref{nb}, \eqref{ng}, and \eqref{bga} into \eqref{no} and then into \eqref{Ione}, we finally get
\begin{equation}\label{I1finbound}
	I_{1}   \le  (\lambda^{-1}D_k+1)\lambda^{-1}D_k^2\W_{1}(\mu,\nu)d\|A\|_{\C^1}
\end{equation}

Now let us bound $I_2$ in \eqref{step0}. We clearly have
\begin{equation}\label{Itwo}
	I_{2}    =| \langle k(x,\cdot)-k(y,\cdot),m_A^\lambda(\cdot;\nu)\rangle|
	\le \| k(x,\cdot)-k(y,\cdot)\|_{\H}\|m_A^\lambda(\cdot;\nu)\|_{\H}
\end{equation}
Note that
\begin{align*}
	&\| k(x,\cdot)-k(y,\cdot)\|_{\H}^2\\
	&\quad =\langle k(x,\cdot)-k(y,\cdot),k(x,\cdot)-k(y,\cdot)\rangle_{\H}\\
	&\quad=k(x,x)-k(x,y)-(k(y,x)-k(y,y))\\
	&\quad =\Bigl( \int_0^1 \d_2 k(x,x+\xi(y-x))\,d\xi\Bigr)^\top(x-y)-\Bigl( \int_0^1 \d_2 k(y,x+\xi(y-x))\,d\xi\Bigr)^\top(x-y)\\
	&\quad =(x-y)^{\top}\Bigl( \int_0^1\int_0^1 \d_1\d_2 k(x+\eta(y-x),x+\xi(y-x))\,d\xi d\eta\Bigr)^\top(x-y),
\end{align*}
with $\partial_{1},\partial_{2}$ denoting the vector of derivatives of $k$ with respect
to the first and second argument, respectively. Recalling Assumption~\ref{AKnew}, we derive
\begin{equation}
	\| k(x,\cdot)-k(y,\cdot)\|_{\H}^2  \leq dD_{k}^{2}\left\vert x-y\right\vert ^{2}. \label{Eknew}
\end{equation}
Further, using \eqref{bga}, we see that
\begin{align*}
	\|m_A^\lambda(\cdot;\nu)\|_{\H}&=\|\beta^\nu\|_{\ell_2(I)}\le \|(B^\nu+\lambda I)^{-1}\|_{\ell_2(I)}\|\gamma^\nu\|_{\ell_2(I)}\le \lambda^{-1}D_k\|A\|_{\C^1}.
\end{align*}
Combining this with \eqref{Eknew} and substituting into \eqref{Itwo}, we get
\begin{equation*}
	I_{2}  \le \sqrt d \lambda^{-1}D_k^2 \|A\|_{\C^1}|x-y|.
\end{equation*}
This, together with \eqref{I1finbound} and \eqref{step0}, finally yields
\begin{equation*}
	|m_A^{\lambda}(x;\mu)-m_A^{\lambda}(y;\nu)| \le C_1 \W_{1}(\mu,\nu) +C_2 |x-y|,
\end{equation*}
where $C_1=(\lambda^{-1}D_k+1)\lambda^{-1}D_k^2d\|A\|_{\C^1}$ and $C_2=\sqrt d \lambda^{-1}D_k^2 \|A\|_{\C^1}$. This completes the proof of the theorem.
\end{proof}

Now we are ready to prove the main results of \cref{sec:mkv}. They would follow from \cref{lip} obtained above.

\begin{proof}[Proof of \cref{prop:exist-reg}]
It follows from \cref{lip},  and the assumptions of the theorem, and the fact that $\W_1$-metric can be bounded from above by the $\W_2$-metric, that the drift and diffusion of \eqref{newsyst1}--\eqref{newsyst3} are Lipschitz and satisfy the conditions of \cite[Theorem~4.21]{CaDe1}. Hence it has a unique strong solution.
\end{proof}

\begin{proof}[Proof of \cref{T:MR2}]
We see that \cref{lip} and the conditions of the theorem implies that all the assumptions of  \cite[Theorem~2.12]{CaDe2} hold
(note that the total state dimension is $2d$ in our case). This implies \eqref{eq:prop-chaos}.
\end{proof}

\begin{proof}[Proof of \cref{L2c}]
Consider the operator $\C^\nu$ in  the
orthonormal basis $(\widetilde{a}_{n})_{n\in J}$ of $\mathcal{H}$. Put
\begin{align*}
	D^\nu:=(  \langle \wt a_{i},\mathcal{C}^{\nu}\wt a_{j}\rangle_{\H}
	\big) _{(i,j)\in J\times J}=(  \langle \wt a_{i},T^{\nu}\wt a_{j}\rangle_{\H}
	\big) _{(i,j)\in J\times J}=(\sigma_j \delta_{ij})_{(i,j)\in J\times J},
\end{align*}
since $\wt a_j$ is an eigenvector of $T^\nu$ with eigenvalue $\sigma_j$. Since $\C^\nu$  is diagonal in this basis, we see  that for $\lambda>0$ one has for $i \in J$
\begin{equation}\label{inverse}
	(\mathcal{C}^{\nu}+\lambda I_{\mathcal{H}})^{-1} \wt a_i=(\sigma_i+\lambda)^{-1}\wt a_i.
\end{equation}

Consider also the function $c^\nu_A$ in this basis. We write for $i\in J$ similar to \eqref{gdef}
\begin{equation*}
	\eta^{\nu}_i:=\langle c^\nu_A,\wt a_i\rangle_\H=\int_{\X\times\X} \wt a_i(x)A(y)\nu(dx,dy),\quad i\in I
\end{equation*}
and we clearly have  $c^{\nu}_A=\sum_{i\in J}\eta_i^\nu \widetilde{a}_{i}$.
Then, using \cref{prop: apcon} and \eqref{inverse} we derive for $\lambda>0$
\begin{align}\label{mlambdanewbas}
	m^{\lambda}_A(\cdot;\nu)&=(\mathcal{C}^{\nu}+\lambda I_{\mathcal{H}})^{-1}c^{\nu}_A
	=\sum_{i\in J}\eta_i^\nu (\mathcal{C}^{\nu}+\lambda I_{\mathcal{H}})^{-1}\widetilde{a}_{i}\nn\\
	&=\sum_{i\in J}\eta_i^\nu (\sigma_i+\lambda)^{-1}\wt a_i.
\end{align}
Next, since $m_A\in\L_2^\nu$, we have
\begin{equation}\label{div}
	P_{\overline{\H}} m_A    =\sum_{i\in J}\left\langle \E_{(X,Y)\sim\nu}\left[
	A(Y)|X=\cdot\right]  ,a_i\right\rangle_{\L_2^\nu}a_i.
\end{equation}
Further, for $i\in J$ we deduce
\begin{align*}
	\bigl\langle \E_{(X,Y)\sim\nu}[A(Y)|X=\cdot] ,a_i\bigr\rangle_{\L_2^\nu}&=
	\int_\X \E_{(X,Y)\sim\nu}\left[
	A(Y)|X=x\right]  a_{i}(x)\nu(dx,\X)\\
	&=\E_{(X,Y)\sim\nu} (a_i(X)\E [A(Y)|X])\\
	&=\E_{(X,Y)\sim\nu} a_i(X) A(Y)\\
	&=\sigma_i^{-1/2}\eta_i^\nu,
\end{align*}
where we used that $\wt a_n=\sqrt{\sigma_n}a_n$. Substituting this into \eqref{div} and combining with \eqref{mlambdanewbas}, we  get
\begin{equation*}
	P_{\overline{\H}} m_A -m_{A}^{\lambda}=\sum_{i\in J} (\eta_i^\nu\sigma_i^{-1}- \eta_i^\nu(\sigma_i+\lambda)^{-1})\wt a_i=
	\sum_{i\in J} \eta_i^\nu\frac{\lambda}{\sigma_i(\sigma_{i}+\lambda)}\wt a_i.
\end{equation*}
Thus
\begin{equation*}
	\bigl\|P_{\overline{\H}} m_A -m_{A}^{\lambda}\bigl\|_{\L_2^\nu}^2=
	\sum_{i\in J} (\eta_i^\nu)^2\frac{\lambda^2}{\sigma_i(\sigma_{i}+\lambda)^2}=
	\sum_{i\in J} \langle m_A ,a_i\rangle_{\L_2^\nu}^2\frac{\lambda^2}{(\sigma_{i}+\lambda)^2},
\end{equation*}
which is \eqref{conver}.
Similarly, recalling \eqref{scpr}, we get
\begin{equation*}
	\bigl\|P_{\overline{\H}} m_A -m_{A}^{\lambda}\bigl\|_{\H}^2=
	\sum_{i\in J} (\eta_i^\nu)^2\frac{\lambda^2}{\sigma_i^2(\sigma_{i}+\lambda)^2}=
	\sum_{i\in J} \langle m_A ,a_i\rangle_{\L_2^\nu}^2\frac{\lambda^2}{\sigma_i(\sigma_{i}+\lambda)^2},
\end{equation*}
which is finite whenever $P_{\overline{\H}} m_A \in \H$, that is, $\sum_{i\in J} \langle m_A ,a_i\rangle_{\L_2^\nu}^2\sigma_i^{-1}<\infty$. This shows \eqref{conver1}. It is easily seen by dominated convergence that the l.h.s. of \eqref{conver} goes to zero, and, in the case $P_{\overline{\H}} m_A \in \H$  the l.h.s. of \eqref{conver1} goes to zero as well.
\end{proof}


\bibliographystyle{alpha}      
\bibliography{GHL}   

\end{document}
